\documentclass[10pt]{amsart}

\usepackage[all]{xy}
\usepackage{amssymb}
\usepackage{eucal}
\usepackage{fullpage}

\usepackage[utf8]{inputenc}
\usepackage[T1]{fontenc}
\usepackage[backend=biber,
            isbn=false,
            doi=false,
            maxbibnames=5,
            giveninits=true,
            style=alphabetic,
            citestyle=alphabetic]{biblatex}
\usepackage{url}
\renewbibmacro{in:}{}
\bibliography{susyqm.bib}

\usepackage{xcolor}
\definecolor{e-mail}{rgb}{0,.40,.80}
\definecolor{reference}{rgb}{.20,.60,.22}
\definecolor{citation}{rgb}{0,.40,.80}

\usepackage[colorlinks=true,
            linkcolor=reference,
            citecolor=citation,
            urlcolor=e-mail]{hyperref}
\usepackage{cleveref}

\newcommand{\cA}{\mathcal{A}}
\newcommand{\cB}{\mathcal{B}}
\newcommand{\cD}{\mathcal{D}}

\newcommand{\cH}{\mathcal{H}}
\newcommand{\cL}{\mathcal{L}}
\newcommand{\cN}{\mathcal{N}}
\newcommand{\cO}{\mathcal{O}}

\newcommand{\cV}{\mathcal{V}}

\newcommand{\bA}{\mathbf{A}}
\newcommand{\C}{\mathbf{C}}
\newcommand{\CC}{\mathbf{C}}
\newcommand{\CP}{\mathbf{CP}}

\newcommand{\bL}{\mathbb{L}}
\newcommand{\bP}{\mathbb{P}}
\newcommand{\RR}{\mathbf{R}}
\newcommand{\R}{\mathbf{R}}

\newcommand{\ZZ}{\mathbf{Z}}
\newcommand{\Z}{\mathbf{Z}}

\renewcommand{\d}{\mathrm{d}}
\newcommand{\D}{\mathrm{D}}
\newcommand{\rH}{\mathrm{H}}
\newcommand{\rP}{\mathrm{P}}
\newcommand{\T}{\mathrm{T}}
\newcommand{\U}{\mathrm{U}}

\newcommand{\sfH}{\mathsf{H}}

\newcommand{\sfR}{\mathsf{R}}
\newcommand{\sfQ}{\mathsf{Q}}
\newcommand{\sfX}{\mathsf{X}}
\newcommand{\sfp}{\mathsf{p}}

\newcommand{\g}{\mathfrak{g}}
\newcommand{\fg}{\mathfrak{g}}
\newcommand{\fX}{\mathfrak{X}}

\DeclareMathOperator{\Ad}{Ad}
\DeclareMathOperator{\ad}{ad}
\newcommand{\Bcc}{\mathcal{B}_{c.c.}}
\newcommand{\Ber}{\mathrm{Ber}}
\newcommand{\BM}{\mathrm{BM}}
\newcommand{\Bun}{\mathrm{Bun}}
\newcommand{\ch}{\mathrm{ch}}
\newcommand{\Conn}{\mathrm{Conn}}
\newcommand{\dCrit}{\mathrm{dCrit}}
\newcommand{\DR}{\mathrm{DR}}
\newcommand{\dvol}{\mathrm{dvol}}
\newcommand{\End}{\mathrm{End}}

\newcommand{\Hom}{\mathrm{Hom}}
\renewcommand{\Im}{\mathrm{Im}}
\newcommand{\Loc}{\mathrm{Loc}}
\newcommand{\Map}{\mathrm{Map}}
\renewcommand{\Re}{\operatorname{Re}}

\newcommand{\Sing}{\mathrm{Sing}}
\newcommand{\SO}{\mathrm{SO}}
\renewcommand{\O}{\mathrm{O}}

\newcommand{\SU}{\mathrm{SU}}
\newcommand{\Sym}{\mathrm{Sym}}

\newcommand{\tr}{\mathrm{tr}}
\newcommand{\triv}{\mathrm{triv}}
\newcommand{\Vect}{\mathrm{Vect}}

\renewcommand{\Bar}{\overline}

\newcommand{\dbar}{{\Bar{d}}}

\newcommand{\ibar}{{\Bar{i}}}
\newcommand{\jbar}{{\Bar{j}}}
\newcommand{\kbar}{{\Bar{k}}}

\newcommand{\mbar}{{\Bar{m}}}

\newcommand{\<}{\langle}
\renewcommand{\>}{\rangle}

\newcommand{\im}{\mathrm{i}}
\def\dbar{{\overline{\partial}}}

\newcommand{\bu}{{\bullet}}
\newcommand{\Sect}{{\rm Sect}}

\newcommand{\defterm}[1]{\textbf{\emph{#1}}}
\newcommand{\define}{\overset{\rm def}{=}}

\newtheorem{thm}{Theorem}[section]

\newtheorem{prop}[thm]{Proposition}
\crefname{prop}{proposition}{propositions}

\newtheorem{cor}[thm]{Corollary}
\newtheorem{lm}[thm]{Lemma}
\newtheorem{conjecture}[thm]{Conjecture}
\newtheorem{proposal}[thm]{Proposal}

\theoremstyle{definition}
\newtheorem{defn}[thm]{Definition}
\newtheorem{notation}[thm]{Notation}

\theoremstyle{remark}
\newtheorem{remark}[thm]{Remark}
\newtheorem{example}[thm]{Example}

\usepackage{todonotes}

\begin{document}
\title{Batalin--Vilkovisky quantization and supersymmetric twists}
\author{Pavel Safronov}
\address{School of Mathematics, University of Edinburgh, Edinburgh, UK}
\email{p.safronov@ed.ac.uk}
\author{Brian R. Williams}
\address{School of Mathematics, University of Edinburgh, Edinburgh, UK}
\email{brian.williams@ed.ac.uk}
\begin{abstract}
We show that a family of topological twists of a supersymmetric mechanics with a K\"ahler target exhibits a Batalin--Vilkovisky quantization. 
Using this observation we make a general proposal for the Hilbert space of states after a topological twist in terms of the cohomology of a certain perverse sheaf. 
We give several examples of the resulting Hilbert spaces including the categorified Donaldson--Thomas invariants, Haydys--Witten theory and the 3-dimensional A-model.
\end{abstract}
\maketitle

\section*{Introduction}

\subsection*{2d A-model and deformation quantization}

Given a symplectic manifold $(M, \omega)$ we may consider its deformation quantization, i.e. a deformation of the commutative algebra of functions $C^\infty(M)$ to a noncommutative algebra. The existence of such deformation quantizations was shown by De Wilde--Lecomte \cite{DeWildeLecomte} and Fedosov \cite{Fedosov} in the smooth context, Nest--Tsygan \cite{NestTsygan} and Polesello--Schapira \cite{PoleselloSchapira} in the complex-analytic context and Bezrukavnikov--Kaledin \cite{BezrukavnikovKaledin} in the algebraic context. Let us also mention the work of Kontsevich \cite{Kontsevich} who proves the existence of a deformation quantization of Poisson manifolds.

For the symplectic manifold $(M, \omega)$ we may also consider a two-dimensional TQFT known as the 2d A-model. Its category of boundary conditions is the Fukaya category of $M$ and the relationship between the Fukaya category and the category of modules over the deformation quantization algebra has a long history \cite{BresslerSoibelman,Tsygan}.

Let us suppose $M$ is a hyperK\"ahler manifold with complex structures $I, J, K$, K\"ahler structures $\omega_I, \omega_J, \omega_K$ and holomorphic symplectic structures $\Omega_I, \Omega_J, \Omega_K$. Consider the real symplectic structure $\omega_J=\Re \Omega_I$ and the B-field $B=\omega_K$. In this case the relationship between the Fukaya category and the category of complex-analytic deformation quantization (DQ) modules on $(M, \Omega_I)$ is expected to be even more tight. Namely, there are no instanton corrections in the Lagrangian Floer homology with boundary on $I$-holomorphic Lagrangians \cite{SolomonVerbitsky}.

A physical explanation of the relationship between the Fukaya category of $(M, \omega_J, B=\omega_K)$ and the deformation quantization of the complex symplectic manifold $(M, \Omega_I)$ has the following two ingredients:
\begin{itemize}
\item Consider a 2d $\cN=(4, 4)$ $\sigma$-model into $M$. It has a $\CP^1\times \CP^1$ family of topological twists corresponding to the $\CP^1\times \CP^1$ family of generalized complex structures obtained from the hyperK\"ahler structure on $M$. One of the topological twists is the 2d B-model into $(M, I)$ while another topological twist is the 2d A-model into $(M, \omega_J)$ (see \cite[Section 4.6]{Gualtieri} for the formula for the family of generalized complex structures). The interpolating family of 2d TQFTs provides a noncommutative deformation of the derived category of coherent sheaves of $(M, I)$ along the holomorphic Poisson bivector $\Omega_I^{-1}$ \cite{KapustinNC,KapustinAbranes}.

\item There is a \emph{canonical coisotropic brane} $\Bcc$ in the A-model which is supported everywhere. As explained in \cite[Section 11]{KapustinWitten} and \cite[Section 2.2]{GukovWitten} the endomorphism algebra of $\Bcc$ provides a deformation quantization of the algebra of holomorphic functions on $(M, I)$. In particular, any other brane $\cB$ gives rise to a DQ module $\Hom(\Bcc, \cB)$ over $\Hom(\Bcc, \Bcc)$.
\end{itemize}

The above perspective on the 2d A-model of a holomorphic symplectic manifold in terms of the deformation quantization allows us to make sense of categories of boundary conditions even when the 2d A-model itself is ill-defined. For example, consider a topological twist of the 4d $\cN=4$ super Yang--Mills theory known as the GL twist (we consider the case $t=1$, $\Psi=0$ in the notation of \cite{KapustinWitten}). Let $G_\C$ be the complexified gauge group. The compactification of the theory on a Riemann surface $\Sigma$ gives the 2d A-model into the Hitchin moduli space of $\Sigma$ with respect to symplectic structure $\omega_K$. In the complex structure $I$ the Hitchin moduli space is given by the cotangent bundle of the moduli space $\Bun_{G_\C}(\Sigma)$ of $G_\C$-bundles on $\Sigma$. In particular, the category of DQ modules is the category of $D$-modules $\cD(\Bun_{G_\C}(\Sigma))$. Even though the 2d A-model into the Hitchin moduli space is ill-defined (as the space is stacky and singular), the category of $D$-modules is well-defined.

As another example, one may consider a topological twist of the 6d $\cN=(1, 1)$ super Yang--Mills theory. The compactification of the theory on a hyperK\"ahler 4-manifold $X$ gives the 2d A-model into the moduli space $\Bun_{G_\C}(X)$ of $G_\C$-bundles on $X$ which is a holomorphic symplectic stack. As before, the ill-defined 2d A-model may be defined rigorously in terms of DQ modules on $\Bun_{G_\C}(X)$.

\subsection*{1d A-model and BV quantization}

The goal of the present paper is to develop an analogous picture by replacing the 2d A-model into a hyperK\"ahler manifold by the 1d A-model (a topological twist of supersymmetric mechanics) into a K\"ahler manifold equipped with a holomorphic superpotential. The procedure of deformation quantization of a symplectic manifold gets replaced with the procedure of \emph{Batalin--Vilkovisky quantization} of a $(-1)$-shifted symplectic manifold which we briefly recall now.

Let $A$ be a graded commutative Poisson algebra with the Poisson bracket of degree $1$ (a $\bP_0$ algebra). A BV operator is a second order differential operator $\Delta$ on $A$ with symbol the Poisson bracket and which is square-zero. A BV operator allows one to deform the differential $\d$ on $A$ to a square-zero differential $\d + \hbar \Delta$. This is parallel to the fact that in the usual deformation quantization the Poisson bracket controls the first-order deformation of the multiplication. We refer to \cref{sect:BV} for more details on BV quantization.

Consider a K\"ahler manifold $M$ equipped with a closed $(1, 0)$ form $\beta$. For instance, $\beta=\partial W$ for a holomorphic superpotential $W\colon M\rightarrow \C$. In this case the supersymmetric mechanics into $M$ admits $\cN=4$ supersymmetry. As in the 2-dimensional case, there is a $\CP^1\times\CP^1$-family of topological twists. In particular, there is a $\CP^1$ family $\sfQ_\hbar$ of supersymmetric twists interpolating between a ``B twist'' for $\hbar=0$ and an ``A twist'' for $\hbar=1$ (see \cref{fig:ABfamily}).

\begin{figure}[h]
\centering
\begin{tikzpicture}
  \shade[ball color = gray!40, opacity = 0.4] (0,0) circle (2cm);
  \draw (0,0) circle (2cm);
  \draw (-2,0) arc (180:360:2 and 0.6);
  \draw[dashed] (2,0) arc (0:180:2 and 0.6);
    \fill[fill=black] (0,2) circle (2pt) node[above]{$Q_B (\hbar = 0)$};
	\fill[fill=black] (1,-1) circle (2pt) node[right]{$Q_A$};
\end{tikzpicture}
\caption{The $\CP^1$ family of twists $\sfQ_\hbar$. At the special point $\hbar = 0$ is the B-twist $Q_B$. Generically we obtain an A twist $Q_A$.}
\label{fig:ABfamily}
\end{figure}
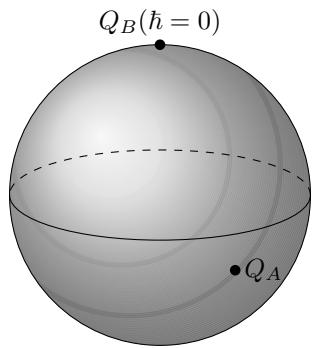

The Hilbert space in the B twist admits an explicit description in terms of (derived) functions on the zero locus $\beta^{-1}(0)$ of $\beta$ (more precisely, one has to consider half-densities on the zero locus). This algebra can be explicitly presented via a Koszul complex. From this presentation it is easy to see that it admits a Poisson bracket of degree $1$ given by the Schouten bracket. Our main observation (see \cref{thm:maintheorem}) is that the family $\sfQ_\hbar$ provides a BV quantization of the zero locus $\beta^{-1}(0)$.

This observation is useful to provide a mathematically rigorous definition of the Hilbert space in the A twist. For instance, it is often the case that one is forced to work with a potential on an infinite-dimensional manifold, so making sense of a Morse--Novikov complex requires hard analysis (for instance, to show that the differential squares to zero). In contrast, the critical locus is often finite-dimensional.

\subsection*{BV quantization and critical cohomology}

Our next observation is that one can in fact provide a topological model of the BV quantization. Namely, the works \cite{BBDJS,BBBBJ} have defined for any $(-1)$-shifted symplectic stack $\sfX$ equipped with an \emph{orientation data}, i.e. the choice of a square root of the canonical bundle $\det(\bL_\sfX)^{\frac 12}$, a perverse sheaf $P_\sfX$ on the underlying classical stack $t_0(\sfX)$. It is expected (see \cref{conj:perverseBV}) that the cohomology of this perverse sheaf gives a model of the BV quantization. Note that the (twisted) BV quantization of $(-1)$-shifted symplectic stacks has been constructed in \cite{PridhamBV}.

Let us give an example of the above. Consider a function $f\colon U\rightarrow \bA^1$ on a smooth affine variety and let $\sfX = \dCrit(f)$ be the \emph{derived critical locus} of $f$. It carries a $(-1)$-shifted symplectic scheme. It has a canonical twisted BV quantization given by the twisted de Rham complex, i.e. the complex of differential forms $\DR(U)(\!(\hbar)\!)$ equipped with the differential $\hbar \d + (\d f)\wedge(-)$. As explained in \cite{SabbahSaito} it is isomorphic to the cohomology of the sheaf of vanishing cycles of $f$ which is exactly what the perverse sheaf $P_\sfX$ is in this case.

The cohomology of the perverse sheaf $P_\sfX$ has been used to define cohomological Hall algebras of quivers with potentials \cite{KontsevichSoibelman}, categorified Donaldson--Thomas invariants \cite{BBBBJ} and complexified Floer homology \cite{AbouzaidManolescu}.

As another example, consider a quasi-smooth derived scheme $\sfX$, so that its underlying scheme $X=t_0(\sfX)$ is a local complete intersection. If $X$ is not smooth, the shifted dualizing complex $\omega_X[-\dim \sfX]$ is not perverse. However, one can define the scheme of singularities $\Sing(X) = t_0(\T^*[-1] \sfX)\xrightarrow{\pi} X$ and the perverse sheaf $P_{\T^*[-1] \sfX}$ on $\Sing(X)$ is such that $\pi_* P_{\T^*[-1]\sfX}\cong \omega_X[-\dim \sfX]$ (see \cite{Kinjo}).

\subsection*{Examples}

We provide many examples where the above ideas give a mathematically rigorous proposal for the space of states in a physical TQFT. In each case we perform the following:
\begin{enumerate}
\item Compute a compactification of a $d$-dimensional TQFT on a $(d-1)$-manifold and rewrite it in terms of (gauged) supersymmetric mechanics for some closed $(1, 0)$-form $\beta$ on an infinite-dimensional K\"ahler manifold $M$ (equipped with a Hamiltonian action).
\item Present $\beta^{-1}(0)$ as a \emph{finite-dimensional} $(-1)$-shifted symplectic stack $\sfX$.
\item Define the space of states to be the cohomology of the perverse sheaf $P_\sfX$.
\end{enumerate}

We observe that in several examples there is a natural grading operator on the space of states manifested in an extra term in the action which is responsible for a ``categorification'' of the corresponding dimensionally reduced theory:
\begin{itemize}
\item The grading in the space of states of the 3d A-model categorifying the Novikov parameter $q$ of the 2d A-model.
\item The grading in the space of states in the Haydys--Witten TQFT categorifying the instanton counting parameter $q$ of the GL twisted 4d $\cN=4$ super Yang--Mills theory.
\end{itemize}

Let us list the relevant examples:
\begin{itemize}
\item \textbf{2d A-model}. Consider a hyperK\"ahler manifold $M$ together with two $I$-holomorphic Lagrangians $L_0, L_1$. The derived intersection $L_0\times_M L_1$ admits a natural $(-1)$-shifted symplectic structure \cite{PTVV}. The space of states in the 2d A-model into $(M, \omega_J)$ compactified on the interval with boundary conditions specified by the Lagrangians $L_0, L_1$ is the cohomology
\[\R\Gamma(L_0\cap L_1, P_{L_0\times_M L_1}).\]
This complex was proposed in \cite{BBDJS,Bussi} as an algebraic model for the Hom spaces in the Fukaya category of a holomorphic symplectic manifold.

\item \textbf{3d A-model}. Consider a hyperK\"ahler manifold $X$ equipped with an isometric $\U(1)$-action rotating the complex structures with fixed points $\pm I$ and acting on $\Omega_I$ with weight $n\neq 0$. The 3d $\cN=4$ supersymmetric $\sigma$-model into $X$ admits a topological twist which gives rise to a 3d TQFT known as the 3d A-model. In the case $X$ is a quaternionic vector space (possibly with an action of a group), this theory has been studied recently in \cite{KapustinVyas,BMG,DGGH,Nakajima} in relation to Coulomb branches.

For a Riemann surface $\Sigma$, the space of sections $\Sect(\Sigma, \fX_\Sigma)$ of the bundle of hyperK\"ahler manifolds $\fX_\Sigma=K^{1/n}_\Sigma\times^{\C^\times} X$ over $\Sigma$ admits a $(-1)$-shifted symplectic structure \cite{GinzburgRozenblyum}. The space of states in the 3d A-model on $\Sigma$ is the cohomology
\[\R\Gamma(\Sect(\Sigma, \fX_\Sigma), P_{\Sect(\Sigma, \fX_\Sigma)}).\]
It admits a grading given by the symplectic volume of the section with respect to $\omega_I$. In the case $X = \T^* Y$, we simply get the Borel--Moore homology
\[\rH^{\BM}_\bullet(\Map(\Sigma, Y))\]
of the moduli space of $I$-holomorphic maps, an answer proposed in \cite{Nakajima}.

\item \textbf{GL twist of the 4d $\cN=4$ super Yang--Mills theory}. The 4d $\cN=4$ super Yang--Mills theory admits a topological twist, the \emph{GL twist}, studied in \cite{KapustinWitten}. It gives rise to a family of 4d TQFTs parametrized by $\Psi\in\CP^1$. Let $G_\C$ be the complexified gauge group. The derived moduli space $\R\Loc_{G_\C}(M)$ of $G_\C$ local systems on a 3-manifold $M$ is a $(-1)$-shifted symplectic stack. The space of states in the GL twist on $M$ for generic $\Psi$ is the cohomology
\[\R\Gamma(\Loc_{G_\C}(M), P_{\R\Loc_{G_\C}(M)})\]
of the perverse sheaf. This may be viewed as a complexified analog of the instanton Floer homology (which gives the space of states in the topologically twisted 4d $\cN=2$ super Yang--Mills theory) and was introduced in \cite{AbouzaidManolescu}.

\item \textbf{Haydys--Witten twist of the 5d $\cN=2$ super Yang--Mills theory}. The 5d $\cN=2$ super Yang--Mills theory admits a topological twist, the \emph{Haydys--Witten twist}, studied in \cite{WittenFivebranes}. Consider the moduli stack $\Bun_{G_\C}(M)$ of principal $G_\C$-bundles on a complex projective surface $M$. The space of states in the Haydys--Witten twist on $M$ is the Borel--Moore homology
\[\rH^{\BM}_\bullet(\Bun_{G_\C}(M))\]
of the moduli stack of $G_\C$-bundles. It admits a grading given by the second Chern character $\int_M \ch_2(P)$ of the $G_\C$-bundle.

\item \textbf{Topological twist of the 7d $\cN=1$ super Yang--Mills theory}. The 7d $\cN=1$ super Yang--Mills theory admits a topological twist on $G_2$ manifolds \cite{AOS}. Consider the derived moduli stack $\R\Bun_{G_\C}(X)$ of $G_\C$-bundles on a projective Calabi-Yau 3-fold $X$. It admits a natural $(-1)$-shifted symplectic structure \cite{PTVV}. The space of states in the topological twist on $X$ is the cohomology
\[\R\Gamma(\Bun_{G_\C}(X), P_{\R\Bun_{G_\C}(X)})\]
of the perverse sheaf on the moduli stack of $G_\C$-bundles. It admits a natural grading by the second Chern character $\int_X \ch_2(P)\wedge \omega$ of the $G_\C$-bundle.
\end{itemize}

We remark on a `chiralization' of the model for BV quantization that we have just proposed.
From the point of view of 2d $\cN=(2,2)$ supersymmetry, the A-model sits in a $\CP^1$-family of twists which at the special point $0 \in \CP^1$ is not a topological twist but a {\em holomorphic} one---this is the so-called half-twist \cite{Kapustin:2005pt}. 
To the data of a K\"ahler manifold $M$ equipped with a closed $(1,0)$-form $\beta$ one can consider the 2d $\cN=(2,2)$ supersymmetric $\sigma$-model into $M$. 
The half-twist produces the holomorphic $\sigma$-model into the derived zero locus $\beta^{-1}(0)$.
The Hilbert space is equipped with a chiral version of a shifted Poisson bracket and the $\CP^1$-family exhibits a chiral BV quantization. 
Upon applying the Zhu algebra construction this recovers the situation above. 

\subsection*{Organization of the paper}

The paper is organized as follows. In \cref{sect:BV} we recall the notion of Batalin--Vilkovisky quantization. We present the BV quantization of odd symplectic manifolds and $(-1)$-shifted symplectic stacks in parallel to emphasize the similarities and differences. For instance, for odd symplectic manifolds there is a canonical line of semidensities which admits a canonical BV operator. In the case of $(-1)$-shifted symplectic stacks the square root of the canonical bundle might not exist and, even if it exists, is not canonical: it is the orientation data. We also conjecture a relationship between the BV quantization and the sheaf of vanishing cycles (\cref{conj:perverseBV}) and describe the latter sheaf in some examples.

In \cref{sect:SUSYmechanics} we describe $\cN=2$ supersymmetric mechanics into a Riemannian manifold equipped with a potential (or a closed one-form $\alpha$). In this section we work on the level of phase spaces and explicitly write down the Hamiltonians for the supersymmetry action. We also show that if the target is K\"ahler and the one-form $\alpha$ is the real part of a closed $(1, 0)$ form $\beta$, the supersymmetry is enhanced to $\cN=4$. By considering a geometric quantization of the phase space we arrive at a description of the Hilbert space of supersymmetric mechanics. In the $\cN=4$ case we describe a family of supercharges $\sfQ_\hbar$ interpolating between the B and the A twist and show that it provides a BV quantization of the B twist (see \cref{thm:maintheorem}). This allows us to formulate a precise proposal for the space of states in the A twist in terms of the cohomology of the perverse sheaf (see \cref{mainproposal}).

In \cref{sect:gaugedmechanics} we describe a gauged version of supersymmetric mechanics and we write down the bosonic actions for both the $\cN=2$ and $\cN=4$ versions. We also formulate a precise proposal for the space of states in the A twist in the presence of gauge symmetries (see \cref{mainproposalstack}).

In \cref{sect:bundlecompactification} we present several results about principal bundles on product manifolds and principal bundles for groups $\Map(N, G)$. These are used in the future sections when we describe compactifications of $G$-gauge theories and rewrite them in terms of $\Map(N, G)$-gauged mechanics.

Finally, in \cref{sect:2dA,sect:3dA,sect:GL,sect:HaydysWitten,sect:G2monopoles} we present the main applications of these ideas which allow us to give a mathematically rigorous definition of the spaces of states in 2d A-model, 3d A-model, the GL twist of the 4d $\cN=4$ super Yang--Mills theory, the Haydys--Witten twist of the 5d $\cN=2$ super Yang--Mills theory and the topological twist of the 7d $\cN=1$ super Yang--Mills theory.

\subsection*{Acknowledgements}
We would like to thank Tudor Dimofte and Sam Gunningham for useful conversations.

\section{Batalin--Vilkovisky quantization}
\label{sect:BV}

Throughout we fix a field $k$. In this section we recall some results on Batalin--Vilkovisky quantization for odd symplectic supermanifolds and $(-1)$-shifted symplectic schemes.

\subsection{Supermanifolds}

Let $V$ be a super $k$-vector space equipped with an odd symplectic structure $\omega$. Consider the $\Z/2$-graded complex $\wedge^\bullet(V^*)$ with the differential $\omega\wedge(-)$. As shown in \cite[Section 3.7]{Manin} and \cite{Severa}, the cohomology of $\wedge^\bullet(V^*)$ is one-dimensional and concentrated in a single degree.

\begin{defn}
Let $(V, \omega)$ be an odd symplectic super vector space. The \defterm{line of semidensities of $V$} is $\Ber^{\frac12}(V) = \rH^\bullet(\wedge^\bullet(V^*), \omega\wedge(-))$.
\end{defn}

The following is a reformulation of \cite[Proposition 3.7]{Manin}.

\begin{defn}
Let $W$ be a super vector space. The \defterm{Berezinian line of $W$} is $\Ber(W) = \Ber^{\frac12}(\Pi\T^* W)$.
\end{defn}

The following is shown in \cite{Severa}.

\begin{lm}
Let $(V, \omega)$ be an odd symplectic super vector space. Then there is a canonical isomorphism of lines $(\Ber^{\frac12}(V))^{\otimes 2}\cong \Ber(V)$.
\end{lm}

The above definitions translate to the global context. Recall the notion of a real or complex supermanifold as in \cite[Chapter 4]{Manin}. For any supermanifold $\sfX$ the supermanifold $\Pi \T^* \sfX$ admits a canonical odd symplectic structure. 
We may thus define the Berezinian line bundle $\Ber_\sfX\rightarrow \sfX$.
If $(\sfX, \omega)$ itself is an odd symplectic supermanifold, the Berezinian admits a canonical square root $\Ber_\sfX^{\frac12}\rightarrow \sfX$ given by the line bundle of semidensities.

\subsection{BV operators}

\begin{defn}
Let $A$ be a commutative $k$-algebra and $M$ an $A$-module. The subspace $\D^{\leq k}(M)\subset \End_k(M)$ of \defterm{differential operators of order $k$} is defined inductively by declaring $\D^{\leq 0}(M) = \End_A(M)$ and $D\in\D^{\leq k}(M)$ if, and only if, $[D, f]\in\D^{\leq(k-1)}(M)$ for every $f\in A$. The algebra $\D(M)$ of \defterm{differential operators} is the union $\D(M) = \bigcup_k \D^{\leq k}(M)$.
\end{defn}

Given a differential operator $D\in\D^{\leq k}(M)$ of order $k$, we may define its symbol $\sigma(D)\colon \Sym^k_A(\Omega^1_A)\otimes_A M\rightarrow M$. We will not need a general definition and will only use the case of second order differential operators.

\begin{defn}
Let $A$ be a $\Z$-graded commutative algebra with a Poisson bracket $\{-, -\}$ of degree $1$ and $M$ is a graded $A$-module. A \defterm{Batalin--Vilkovisky (BV) operator on $M$} is a degree $1$ square-zero $k$-linear endomorphism $\Delta\colon M\rightarrow M$ satisfying
\[\Delta(fgm) = \Delta(f)gm + (-1)^{|f|} f\Delta(g) m - (-1)^{|f|+|g|} fg \Delta(m) + \{f, g\}m.\]
\end{defn}

In other words, a BV operator on $M$ is a degree 1 square-zero second-order differential operator on $M$ whose symbol is given by the Poisson bracket. The same definition works for $\Z/2$-graded algebras and modules.

For an odd symplectic supermanifold $(\sfX, \omega)$, the algebra of smooth functions $C^\infty(X)$ carries a canonical odd Poisson bracket. The following is shown in \cite{Khudaverdian,Severa}.

\begin{prop}
Let $(\sfX,\omega)$ be an odd symplectic supermanifold. There is a canonical BV operator $\Delta$ on $C^\infty(\sfX; \Ber^{\frac12}_{\sfX})$.
\label{prop:KhBVquantization}
\end{prop}

\subsection{BV quantization}
\label{sect:BVquantization}

The following definition is a version of \cite[Definition 2.4.1.1]{CostelloGwilliam2}.

\begin{defn}
Let $(A, \d_A)$ be a differential graded commutative algebra equipped equipped with a $\d$-closed Poisson bracket of degree $1$ and $(M, \d_M)$ a differential graded $A$-module. A \defterm{Batalin--Vilkovisky (BV) quantization of $M$} is a degree $1$ square-zero $k[\![\hbar]\!]$-linear differential operator $\Delta_\hbar = \d_M + \sum_{n=1}^\infty \Delta^{(n)} \hbar^n$ on $M[\![\hbar]\!]$ satisfying the following conditions:
\begin{itemize}
\item $\Delta^{(1)}$ has order $2$ and symbol the Poisson bracket on $A$.

\item $\Delta^{(n)}$ for $n\geq 2$ has order $n$.
\end{itemize}
\label{def:BVquantization}
\end{defn}

The following example explains a relationship between BV quantizations and BV operators.

\begin{example}
Let $A$ be a graded commutative algebra equipped with a degree $1$ Poisson bracket and $M$ a graded $A$-module equipped with a BV operator $\Delta$. Then $\Delta_\hbar = \hbar\Delta$ is a BV quantization of $M$.
\end{example}

\Cref{def:BVquantization} can be generalized to the homotopical context where we consider commutative dg algebras equipped with a homotopy $\bP_0$-structure (i.e. a degree $1$ Poisson bracket satisfying the Jacobi identity up to coherent homotopy), see \cite[Definition 1.12]{PridhamBV}. In this definition we require $\Delta^{(n)}$ to have order $n+1$ such that the total symbol of $\Delta_\hbar$ recovers the homotopy $\bP_0$-structure on $A$.

Let us describe two examples relevant for the future sections.

\begin{example}
Let $X$ be a smooth manifold and consider the graded commutative algebra
\[A = \Gamma(X, \Sym(\T_X[1]))\]
of polyvector fields. The Schouten bracket endows it with a degree $1$ Poisson bracket. There is a natural $A$-module structure on the space of differential forms
\[M = \Gamma(X, \Sym(\T^*_X[-1]))\]
given by contraction. Then the de Rham differential on $M$ is a second-order differential operator with symbol the Schouten bracket \cite{WittenBV}. In other words, it provides a BV quantization of the $A$-module $M$.
\label{ex:realdifferentialformsBV}
\end{example}

\begin{example}
Let $X$ be a complex manifold and consider the graded commutative algebra
\[A = \Gamma(X, \Sym(\Omega^{0, 1}_X[-1])\otimes \Sym(\T^{1, 0}_X[1]))\]
equipped with the differential $\Bar{\partial}$. As before, the Schouten bracket gives a degree $1$ Poisson bracket on $A$. Consider the $A$-module
\[M = \Gamma(X, \Sym(\Omega^{0, 1}_X[-1])\otimes \Sym(\Omega^{1, 0}_X[-1]))\]
of differential forms. Then $\overline{\partial} + \hbar\partial$ provides a BV quantization of $M$.
\label{ex:complexdifferentialformsBV}
\end{example}

\subsection{BV quantization of shifted symplectic stacks}

Let us now explain the notion of BV quantization of shifted symplectic stacks. Recall that \cite{PTVV} have introduced the notion of a $(-1)$-shifted symplectic stack, i.e. a derived Artin stack $\sfX$ equipped with a (homotopy) symplectic structure $\omega$ of degree $(-1)$. The following is \cite[Theorem 3.2.4]{CPTVV} and \cite[Theorem 3.33]{PridhamPoisson}.

\begin{prop}
Let $\sfX$ be a $(-1)$-shifted symplectic stack. Then there is a canonical homotopy $\bP_0$ structure on the commutative dg algebra of global functions $\R\Gamma(\sfX, \cO)$.
\end{prop}

In particular, for any line bundle $\cL$ on a $(-1)$-shifted symplectic stack $\sfX$ we may define a BV quantization of $\cL$ using \cref{def:BVquantization} applied to the $\R\Gamma(\sfX, \cO)$-module $\R\Gamma(\sfX, \cL)$.

Let us recall that any $(-1)$-shifted symplectic scheme $\sfX$ is quasi-smooth; in particular, Gorenstein. In other words, the dualizing sheaf $\omega_X$ is a line bundle. The following result follows from \cite[Proposition 4.6]{PridhamBV}.

\begin{prop}
Let $\sfX$ be a $(-1)$-shifted symplectic scheme equipped with a square root $\omega_X^{\frac12}$ of the dualizing sheaf. Then there is a canonical BV quantization of $\omega_X^{\frac12}$.
\label{prop:canonicalBVquantization}
\end{prop}

\begin{remark}
As opposed to the case of odd symplectic supermanifolds, $(-1)$-shifted symplectic schemes do not have a canonical choice of the square root $\omega_X^{\frac12}$. Apart from this difference, \cref{prop:canonicalBVquantization} is an analog of \cref{prop:KhBVquantization} in the setting of shifted symplectic schemes.
\end{remark}

\subsection{Perverse sheaf on shifted symplectic stacks}
\label{sect:perversesheaf}

Recall that for a $(-1)$-shifted symplectic scheme $\sfX$ the dualizing sheaf $\omega_\sfX$ is a line bundle. In fact, by \cite[Lemma 3.7]{HalpernLeistner} one has $\omega_\sfX\cong \det(\bL_\sfX)$, the determinant of the cotangent complex of $\sfX$. A $(-1)$-shifted symplectic Artin stack $\sfX$ is no longer quasi-smooth. Nevertheless, it turns out the correct replacement for the dualizing complex $\omega_\sfX$ is the canonical bundle, i.e. the determinant of the cotangent complex $\det(\bL_\sfX)$. The following notion was introduced in \cite{BBBBJ}.

\begin{defn}
Let $(\sfX, \omega)$ be a $(-1)$-shifted symplectic Artin stack. An \defterm{orientation data} on $\sfX$ is the choice of a square root $\det(\bL_\sfX)^{\frac12}$.
\end{defn}

The orientation data was used in \cite[Theorem 1.3]{BBBBJ} to construct a canonical perverse sheaf associated to a $(-1)$-shifted symplectic stack.

\begin{prop}
Let $(\sfX, \omega)$ be a $(-1)$-shifted symplectic Artin stack equipped with an orientation data. Then there is a canonical perverse sheaf $P_\sfX$ of $k$-vector spaces on the underlying classical stack $t_0(\sfX)$.
\end{prop}

\begin{example}
Let $M$ be a smooth complex algebraic symplectic variety and $L_0, L_1\subset M$ two smooth algebraic Lagrangains. Then the derived intersection $\sfX = L_0\times_M L_1$ carries a natural $(-1)$-shifted symplectic structure \cite{PTVV}. Square roots of canonical bundles on $L_i$ give rise to an orientation data on $\sfX$. In this case the perverse sheaf on $t_0(\sfX) = L_0\cap L_1$ has been constructed in \cite{Bussi} without an appeal to derived algebraic geometry. Moreover, it was constructed for complex analytic varieties, i.e. the choice of the algebraic structure in this case is irrelevant.
\label{ex:Lagrangianintersection}
\end{example}

\begin{example}
Let $\sfX$ be a quasi-smooth derived Artin stack. Then $\T^*[-1] \sfX$ is a $(-1)$-shifted symplectic stack equipped with a canonical orientation data. The underlying classical stack $t_0(\T^*[-1] \sfX)$ is known as the stack of \emph{singularities} of $\sfX$ (see \cite[Section 8]{ArinkinGaitsgory}). By \cite[Theorem 4.14]{Kinjo} we have
\[\rH^\bullet(t_0(\T^*[-1] \sfX), P_{\T^*[-1] \sfX})\cong \rH^{\BM}_{\dim(\sfX)-\bullet}(t_0(\sfX)),\]
where on the right we consider the (shifted) Borel--Moore homology of the underlying classical stack $t_0(\sfX)$.
\label{ex:cotangent}
\end{example}

We may now formulate a conjecture relating the perverse sheaf $P_\sfX$ and BV quantization.

\begin{conjecture}
Let $(\sfX, \omega)$ be a $(-1)$-shifted symplectic Artin stack equipped with an orientation data $\det(\bL_\sfX)^{\frac12}$. There is a canonical BV quantization $\Delta_\hbar$ of $\det(\bL_\sfX)^{\frac12}$ and a quasi-isomorphism
\[\R\Gamma(t_0(\sfX), P_\sfX)(\!(\hbar)\!)\cong (\R\Gamma(\sfX, \det(\bL_\sfX)^{\frac12})(\!(\hbar)\!), \Delta_\hbar)\]
of chain complexes of $k(\!(\hbar)\!)$-vector spaces.
\label{conj:perverseBV}
\end{conjecture}

\begin{example}
Let $U$ be a smooth affine variety and $f\colon U\rightarrow \bA^1$ a function. Consider the \emph{derived critical locus} $\sfX = \dCrit(f)$ of the function $f$ (see \cite{Vezzosi} for what this means). It carries a natural $(-1)$-shifted symplectic structure. Let $\pi\colon \sfX\rightarrow U$ be the natural projection. We have a fiber sequence
\[\pi^* \T^*_U\longrightarrow \bL_{\sfX}\longrightarrow \bL_{\sfX / U}\cong \pi^* \T_U[1].\]
Taking the determinant, we obtain an isomorphism $\det(\bL_\sfX)\cong \pi^*\det(\T^*_U)^{\otimes 2}$. In particular, $\pi^*\det(\T^*_U)$ provides an orientation data on $\sfX$. We may identify
\[\R\Gamma(\sfX, \pi^*\det(\T^*_U))\cong (\Gamma(U, \Sym(\T^*_U[-1]))[\dim U], \d f\wedge(-)).\]
In this case
\[\Delta_\hbar = \d f\wedge(-) + \hbar \d\]
provides a BV quantization known as the twisted de Rham complex of $(U, f)$. The above conjecture in this case has been proven in \cite[Proposition 4.9]{PridhamBV}.
\end{example}

\begin{remark}
In the future sections we will encounter situations where the space is obtained as a critical locus of a function on an infinite-dimensional space. In this case there is no canonical orientation data. The cohomology $\R\Gamma(t_0(\sfX), P_\sfX)$ will then be a replacement for the ill-defined BV quantization.
\end{remark}

\section{Supersymmetric mechanics}
\label{sect:SUSYmechanics}

\subsection{Supersymmetry algebras}

In this section we recall some facts about supersymmetry algebras that we will use. We refer to \cite{DeligneFreed} for more details.

\begin{defn}
Let $W$ be a real vector space equipped with a nondegenerate symmetric bilinear pairing $(-, -)$. The \defterm{1d supertranslation algebra} $\g_W$ is the real super Lie algebra $\g = \Pi W\oplus \R\cdot H$ with the relations
\begin{align*}
[H, Q] &= 0 \\
[Q_1, Q_2] &= (Q_1, Q_2) H
\end{align*}
for every $Q_1, Q_2, Q\in W$.
\end{defn}

We will be interested in the following two examples.

\begin{example}
Consider $W$ the two-dimensional real vector space with a metric of signature $(1, 1)$. The corresponding super Lie algebra is known as the \defterm{1d $\cN=2$ supertranslation algebra}. It has odd generators $Q_1, Q_2$ and an even central generator $H$ with the relations
\begin{align*}
[Q_1, Q_1] &= H \\
[Q_2, Q_2] &= -H \\
[Q_1, Q_2] &= 0.
\end{align*}
\end{example}

\begin{example}
Consider $W$ the four-dimensional real vector space with a metric of signature $(2, 2)$. The corresponding super Lie algebra is known as the \defterm{1d $\cN=4$ supertranslation algebra}. Its complexification has odd generators $Q_1^\pm, Q_2^\pm$ and a central even generator $H$ with the nontrivial brackets
\begin{align*}
[Q_1^+, Q_1^-] &= \frac{1}{2}H \\
[Q_2^+, Q_2^-] &= -\frac{1}{2}H
\end{align*}
The real structure is given on the generators by $(Q_\alpha^\pm)^* = Q_\alpha^\mp$ ($\alpha=1, 2$) and $H^* = H$.
\end{example}

\begin{remark}
In both cases $\cN$ refers to the dimension of the odd part. Since we are talking about the real supertranslation algebras, one needs in addition to fix the signature of the metric on $W$, which is implicit in the notation.
\end{remark}

\begin{remark}
There is an embedding of the 1d $\cN=2$ supertranslation aglebra into the 1d $\cN=4$ supertranslation algebra given by the formulas $Q_1 = Q_1^+ + Q_1^-$ and $Q_2 = Q_2^+ + Q_2^-$.
\label{rmk:N=2intoN=4}
\end{remark}

Observe that $\O(W)$ acts on the 1d supertranslation algebra $\g_W$ by outer automorphisms.

\begin{defn}
Let $V$ be a (real or complex) $\Z/2$-graded representation of the 1d supertranslation algebra $\g_W$. In addition, suppose $V$ carries an action of a subgroup $G_R\subset \O(W)$ compatible with the action of $\g_W$. We say $G_R$ is the \defterm{$R$-symmetry group} of the representation. When $V$ is a complex $\g_W$-representation, we assume that the $G_R$-action extends to an action of the complexification $G_{R,\C}$.
\end{defn}

\begin{defn}
A real (complex) \defterm{twisting supercharge} $Q\in \g_W$ ($Q\in\g_W\otimes_\R\C$) is a nonzero odd square-zero element. Given a real (complex) $\g_W$-representation $V$ with $R$-symmetry group $G_R$, the \defterm{grading} of $Q$ is a choice of a subgroup $\R^\times\subset G_R$ ($\C^\times\subset G_{R,\C}$) satisfying the following properties:
\begin{itemize}
\item The induced $\Z/2$-grading on $\g_W$ coincides with the original $\Z/2$-grading.
\item With respect to the induced $\Z$-grading $Q$ has weight $1$.
\end{itemize}
\end{defn}

We are now ready to define the notion of twisting of $\g_W$-representations.

\begin{defn}
Let $V$ be a real (complex) $\g_W$-representation with $R$-symmetry group $G_R$. Suppose $Q$ is a real (complex) twisting supercharge equipped with a grading. The \defterm{twist} of $V$ is the $\Z$-graded complex whose underlying graded vector space is $V$ and differential is $Q$.
\end{defn}

As we will only discuss the $\cN=2$ and $\cN=4$ supertranslation algebras, let us explicitly describe the collection of square-zero supercharges in these algebras.

\begin{prop}
The set of equivalence classes of square-zero supercharges up to scale in the 1d supertranslation algebra $\g_W$ is the set of null lines in $W$.
\begin{itemize}
\item The set of real square-zero supercharges up to scale in the 1d $\cN=2$ supertranslation algebra consists of two points $Q_1 + Q_2$ and $Q_1 - Q_2$.

\item The set of complex square-zero supercharges up to scale in the 1d $\cN=4$ supertranslation algebra is $\CP^1\times \CP^1$. Using the homogeneous coordinates $[a:b],[c:d]$ of $\CP^1\times \CP^1$ they are given by
\[bd Q_1^+ + ac Q_1^- + ad Q_2^+ + cb Q_2^-.\]
\end{itemize}
\label{prop:squarezerosupercharges}
\end{prop}
\begin{proof}
The first two claims are obvious. For the last claim, the square-zero supercharges up to scale are given by a smooth quadric in $\CP^3$. But it is well-known that it is given by the Segre embedding $\CP^1\times \CP^1\subset \CP^3$ described above.
\end{proof}

\begin{defn}
The element $Q_A = Q_1 + Q_2$ of the 1d $\cN=2$ supertranslation algebra is the \defterm{A twisting supercharge}. The embedding $\R^\times\subset G_R=\O(1, 1)$ endows it with a grading. Given a representation of the 1d $\cN=2$ supertranslation algebra, its \defterm{A twist} is the twist with respect to $Q_A$.
\end{defn}

\begin{defn}
The element $Q_B = Q_1^- + Q_2^-$ of the 1d $\cN=4$ supertranslation algebra is the \defterm{B twisting supercharge}. The embedding $\R^\times\subset \O(1, 1)\subset G_R=\O(2, 2)$ endows it with a grading. Given a representation of the 1d $\cN=4$ supertranslation algebra, its \defterm{B twist} is the twist with respect to $Q_B$.
\end{defn}

Using the embedding of the 1d $\cN=2$ supertranslation algebra into the 1d $\cN=4$ supertranslation algebra provided by \cref{rmk:N=2intoN=4} we will also talk about the \emph{A twist} of a representation of a 1d $\cN=4$ supertranslation algebra.

\subsection{Symplectic supermanifolds} \label{sec:supersymplectic}

The phase space of a supersymmetric mechanical system is described by a symplectic supermanifold (here we are considering \emph{even} symplectic structures on supermanifolds).
It turns out that there is a down-to-earth description of symplectic supermanifolds in terms of ordinary (non-super) geometry that we briefly recollect.
For more details we refer to \cite{Rothstein}.

Recall that every supermanifold is (non-canonically) equivalent to the total space of a $\ZZ/2$-graded vector bundle over an ordinary manifold.
Similarly, any symplectic supermanifold is equivalent to one which is in a certain standardized form.

To describe this standardized form, fix a tuple of data $(X, \omega_0, V, g, \nabla)$ where
\begin{itemize}
  \item $X$ is an ordinary symplectic manifold with symplectic form $\omega_0 \in \Omega^2(X)$, and
  \item $V$ is a vector bundle on $X$ equipped with a metric $g$ and a connection $\nabla$ which preserves $g$.
\end{itemize}

This data determines a symplectic structure $\omega$ on the supermanifold given by the total space of the bundle $\Pi V$ over $X$
\[
  \sfX = {\rm Tot}(\Pi V)
\]
which we can describe as follows.
Let
\[
  R \in \Omega^2(X; \End(\cV))
\]
the curvature of the connection $\nabla$ and define
\[
  {\rm R} \in \Omega^2(X; \wedge^2 \cV^*)
\]
to be the contraction of $R$ with the metric $g$.
In local coordinates one has ${\rm R}_{ijab} = g_{bc} R_{ij a}^c$ where $\{R_{ija}^c\}$ are the components of the curvature $R$.

There is a short exact sequence of $C^\infty_X$-modules
\[
  0\longrightarrow (\wedge^\bullet \cV^*) \otimes \cV \longrightarrow \Vect(\sfX) \longrightarrow (\wedge^\bullet \cV^*) \otimes \Vect(X)\longrightarrow 0
\]
The connection on $V$ defines a splitting of this short exact sequence and hence determines an isomorphism of $C^\infty_X$-modules $\Vect(\sfX) \cong_\nabla (\wedge \cV^*) \otimes (\cV \oplus \Vect(X))$.

Using this splitting, one defines the following two-form $\omega$ on the supermanifold $\sfX$ by the formulas:
\begin{equation}\label{eqn:rothstein}
  \begin{array}{cclcccc}
    \omega (\mu, \nu) & = & \omega_0(\mu,\nu) + \frac12 {\rm R}(\mu,\nu) & \mbox{for} & \mu,\nu \in \Vect(X) \\
    \omega (\phi, \psi) & = & g(\phi , \psi) & \mbox{for} & \phi,\psi \in \cV \\
    \omega (\phi, \mu) & = & 0 .
  \end{array}
\end{equation}

The main result of \cite{Rothstein} can be summarized in the following way.

\begin{thm}\label{thm:rothstein}
  Let $(X,\omega_0, V, g, \nabla)$ be the tuple of data as defined above.
  The two-form $\omega$ defined in \eqref{eqn:rothstein} is an even symplectic form on the supermanifold $\sfX = {\rm Tot}(\Pi V)$.
  Moreover, any symplectic supermanifold is equivalent to one of this form.
\end{thm}

\subsection{Supersymmetric mechanics}

In this section we introduce the phase space of supersymmetric classical mechanics as in \cite{WittenQM}.

We fix the following data:
\begin{itemize}
  \item A Riemannian manifold $(M,g)$.
  \item A pair of closed one-forms $\alpha, a \in \Omega^1(M)$.
\end{itemize}

\begin{notation}
  In this section $\nabla$ denotes the Levi--Civita connection with respect to the metric $g$ and $R$ denotes the curvature tensor.
\end{notation}

\begin{defn}\label{defn:1dN=2}
  The \defterm{phase space of supersymmetric mechanics} is the supermanifold
  \begin{equation}\label{eqn:N=2phase}
  \sfX \define {\rm Tot}\left(\Pi (\pi^*\T_M\oplus \pi^*\T_M) \right)
\end{equation}
where $\pi\colon \T^* M \to M$ is the projection.
\end{defn}



The graded ring of functions on $\sfX$ is
\begin{align*}
  C^\infty(\sfX) & = \Gamma(M, \Sym(\T_M) \otimes \wedge^\bu (\T_M^* \oplus \T_M^*)).
\end{align*}



\begin{cor} \label{cor:N=1symp}
  By \cref{thm:rothstein} the tuple of data
  \[
    (X, \omega_0, V, g, \nabla) = (\T^*M , \omega_{\rm std}, \pi^*(\T_M \oplus \T_M), (g \oplus -g), (\nabla \oplus \nabla)) .
    \]
  defines a symplectic structure $\omega \in \Omega^2(\sfX)$ on $\sfX$.
\end{cor}


On the supermanifold $\sfX$ there exists a tautological one-form $\sfp\in\Gamma(\sfX, \pi^* \T^*_M)$ (the Liouville one-form on $\T^* M$) and two tautological odd vector fields $\vartheta_1, \vartheta_2\in\Pi\Gamma(\sfX, \pi^* \T_M)$ corresponding to the two copies of the tangent bundle in the definition of $\sfX$.

By \cref{cor:N=1symp} the phase space of the supersymmetric mechanics is equipped with a symplectic form $\omega$. Using these tautological sections, one can write this symplectic form as
\[
  \omega = \d \sfp + \left(R \vartheta_1, \vartheta_1 \right) - \left(R \vartheta_2, \vartheta_2\right) + \left(\nabla \vartheta_1, \nabla \vartheta_1\right) - \left(\nabla \vartheta_2, \nabla \vartheta_2\right).
\]

\begin{prop} \label{prop:1dN=2}
  Consider the pair of odd functions
  \begin{align*}
    \sfQ_1 & = \<\sfp, \vartheta_1\> + \<\alpha, \vartheta_2\> \\
    \sfQ_2 & = \<\sfp, \vartheta_2\> + \<\alpha, \vartheta_1\>.
  \end{align*}
  and the even function
\[
\sfH = \frac12 (\sfp, \sfp) + \frac13 \epsilon^{ab} \epsilon^{cd} (\vartheta_a, R(\vartheta_b, \vartheta_c) \vartheta_d ) + \epsilon^{ab} \<\nabla \alpha , \vartheta_a \otimes \vartheta_b\> - \frac12 (\alpha, \alpha)
\]
  These satisfy the following Poisson brackets
\[
    \{\sfQ_1, \sfQ_1\} = \sfH,\quad
    \{\sfQ_2, \sfQ_2\} = -\sfH , \quad
    \{\sfQ_1, \sfQ_2\} = 0 .
\]
  In other words, the functions $\sfQ_\alpha, \sfH$ determine a Hamiltonian action on $\sfX$ by the 1d $\cN=2$ supertranslation algebra.

  Additionally, the even function $\sfR = \frac12 (\vartheta_1 + \vartheta_2, \vartheta_1 + \vartheta_2)$ is a Hamiltonian for the R-symmetry group $G_R  = \SO(1,1)$ acting by rotations on $\T^*_M\oplus \T^*_M$.
\end{prop}

\begin{proof}

Choose local coordinates $\{q^i\}$ for $M$ and denote by $\{p_i\}$ the corresponding frame of the cotangent bundle. Then we have $\sfp = p_i \d q^i$. Let $\{\theta_1^i\}$ and $\{\theta_2^j\}$ be the corresponding frames of the tangent bundles in \eqref{eqn:N=2phase}. Then $\vartheta_\alpha = \theta^i_\alpha \frac{\partial}{\partial q^i}$ for $\alpha=1,2$.
Note that the coordinates $q^i, p_j$ are even and the coordinates $\theta_1^i, \theta_2^j$ are odd.
In these local coordinates the symplectic form $\omega$ on $\sfX$ reads
\[
  \omega = \d p_i \d q^i + \frac12 {\rm R}_{ijk\ell} (\theta_1^k \theta_1^\ell - \theta_2^k \theta_2^\ell) \d q^i \d q^j + g_{ij} (\nabla \theta_1^i \nabla \theta_2^i - \nabla \theta^i_2 \nabla \theta^j_2) .
\]
Here $\nabla \theta^i_\alpha = \d \theta^i_\alpha - \Gamma^i_{jk} \theta^j_\alpha \d q^k$ denotes the covariant derivative for $a =1,2$.

We record the Poisson brackets read off from the above formula of the symplectic form:
\[
  \begin{array}{cclcccl}
    \{q^i, q^j\} & = & 0 & , & \{p_i, p_j\} & = & - {\rm R}_{ijk\ell} \left(\theta_1^k \theta_1^\ell - \theta_2^k \theta^\ell_2\right) \\
    \{p_i, q^j\} & = & \delta_i^j & , & \{\theta_a^i, p_j\} & = & \Gamma_{jk}^i \theta_a^k \\
    \{\theta_a^i, q^j\} & = & 0 & , & \{\theta_a^i, \theta_b^j\} & = & \frac12 (-1)^{a +1} \delta_{a b} g^{ij} \\
  \end{array}
\]
In coordinates the supercharges $\sfQ_{a}$ read
  \begin{align*}
    \sfQ_1 & = p_i \theta_1^i + \alpha_i \theta_2^i\\
    \sfQ_2 & = p_i \theta_2^i + \alpha_i \theta_1^i
  \end{align*}
  where we have written the one-form $\alpha$ as $\alpha_i \d q^i$.

  The proof proceeds with a direct calculation of $\{\sfQ_a, \sfQ_b\}$ using these local coordinate descriptions.
  First, consider the bracket $\{\sfQ_1, \sfQ_1\}$.
  This commutator splits up into a sum of four commutators:
  \begin{itemize}
    \item $\{p_i \theta_1^i, p_j \theta_1^j\} = 
    {\rm R}_{ijk\ell} \left(- \theta_1^k \theta_1^\ell + \theta_2^k \theta_2^\ell \right)\theta_1^i \theta_1^j + \left(\Gamma_{ik}^{j} \theta_1^k \theta_1^i p_j + \Gamma^{i}_{jk} \theta_1^k \theta_1^j p_i \right) + \frac12 g^{ij} p_i p_j .$
    \item $\{p_i \theta_1^i , \alpha_j \theta_2^j\} = 
        (\partial_i \alpha_j) \theta^i_1 \theta^j_2 - \Gamma_{ik}^j \alpha_j \theta_1^i\theta_2^k .$
    \item $\{\alpha_i \theta_2^i, p_j \theta_1^j\} =
        (\partial_j \alpha_i) \theta^j_1 \theta^i_2 - \Gamma_{jk}^i \alpha_i \theta_1^j\theta_2^k.$
    \item $\{\alpha_i \theta_2^i, \alpha_j \theta_2^j \} =
        - \frac12 g^{ij} \alpha_i \alpha_j .$
\end{itemize}

The term ${\rm R}_{ijk\ell} \theta_1^i \theta_1^j \theta_1^k \theta_1^\ell$ is identically zero by symmetries of the Riemann tensor ${\rm R}_{i[jk\ell]} = 0$.
The third and fourth terms in the first item cancel with one another due to the antisymmetry of the connection one-form.


  Combining the remaining terms we obtain
  \[
    \{\sfQ_1, \sfQ_1\} = \frac12 g^{ij} p_i p_j + {\rm R}_{ijk\ell} \theta_1^i \theta_1^j \theta_2^k \theta_2^\ell + 2 \left(\partial_i \alpha_j - \Gamma_{ij}^k \alpha_k \right) \theta_1^i \theta_2^j - \frac12 g^{ij} \alpha_i \alpha_j .
  \]
Using symmetries of the Riemann tensor again again one can show that the second term is equal to
  \[
    \frac13 \epsilon^{ab}\epsilon^{cd}{\rm R}_{ijk\ell} \theta_a^i \theta_b^j \theta_c^k \theta_d^\ell .
  \]
  We conclude that $\{\sfQ_1,\sfQ_2\} = \sfH$ as desired.
  The relation $\{\sfQ_2, \sfQ_2\} = - \sfH$ is proved similarly.

  Finally, we show $\{\sfQ_1, \sfQ_2\}=0$.
  Again, there are four terms in the expansion of this bracket.
   \begin{itemize}
    \item $\{p_i \theta_1^i, p_j \theta_2^j\}=
    {\rm R}_{ijk\ell}\left(-\theta_1^k \theta_1^\ell + \theta_2^k \theta_2^\ell \right)\theta_1^i \theta_2^j + \left(\Gamma_{ik}^{j} \theta_1^i \theta_2^k p_j + \Gamma^{i}_{jk} \theta_1^k \theta_2^j p_i \right) .$
    \item $\{p_i \theta_1^i , \alpha_j \theta_1^j\}=
        - \partial_i \alpha_j \theta^i_1 \theta^j_1 + \Gamma_{ik}^j \alpha_j \theta_1^i\theta_1^k + \frac12 g^{ij} p_i \alpha_j .$
    \item $\{\alpha_i \theta_2^i, p_j \theta_2^j\}=
        -\partial_j \alpha_i \theta^j_2 \theta^i_2 + \Gamma_{jk}^i \alpha_i \theta_2^j\theta_2^k - \frac12 g^{ij} p_j \alpha_i.$
    \item $\{\alpha_i \theta_1^i, -\alpha_j \theta_2^j \}=0$.
\end{itemize}

The first two terms in the first item are individually zero by using ${\rm R}_{i[jk\ell]} = 0$.
The first terms in the second and third items are individually zero since $\alpha$ is closed $\d \alpha = 0$.
The second terms in the second and third items are individually zero since the connection is torsion-free $\Gamma^i_{[jk]} = \Gamma^i_{jk} - \Gamma^i_{kj} = 0$.
The remaining terms clearly cancel.
Thus, $\{\sfQ_1, \sfQ_2\} = 0$ as desired.
\end{proof}

\subsection{K\"ahler case}
\label{sect:Kahlermechanics}

In this section we specialize to the case where $M$ is equipped with a K\"ahler structure.
We show that in this case the supersymmetric mechanics from the previous section has an enhanced supersymmetry.

Let us fix the following data:
\begin{itemize}
  \item A K\"{a}hler manifold $(M, g, J)$.
  \item A pair of closed one-forms $\alpha,a \in \Omega^1(M)$ whose $(1,0)$ parts $\beta = \alpha^{1,0}$ and $a^{1, 0}$ are closed.
\end{itemize}

Let $\sfX$ be the phase space of supersymmetric mechanics from \cref{defn:1dN=2} defined using the Riemannian structure on $M$. Let $\T_M \otimes_\RR \CC = \T_M^{1,0} \oplus \T_M^{0,1}$ be the decomposition of the tangent bundle using the complex structure $J$, and similarly for the complexified cotangent bundle.
The graded ring of complex-valued functions on $\sfX$ is
\[C^\infty(\sfX; \C) = \Gamma\left(M, \Sym(\T_M^{1,0}) \otimes \Sym(\T_M^{0,1}) \otimes \left(\wedge^\bu (\T^{*1,0}_M) \otimes \wedge^\bu(\T_M^{*0,1}) \right)^{\otimes 2}\right).\]
It carries a natural real structure whose real subspace is $C^\infty(\sfX; \R)$.

\begin{notation}
Let $\sfp = \sfp^{1, 0} + \sfp^{0, 1}$ be the decomposition of the tautological one-form on $\sfX$ according to type and, similarly, $\vartheta_\alpha = \vartheta^{1,0}_\alpha + \vartheta^{0,1}_\alpha$.
\end{notation}

\begin{prop}\label{prop:1dN=4}
  Consider the tuple of odd functions
  \[
    \begin{array}{cccccccc}
    \sfQ_1^+ & = & \<\sfp^{1,0} , \vartheta^{1,0}_1\> + \<\Bar{\beta}, \vartheta^{0,1}_2\> & , & \sfQ_1^- & = & \<\sfp^{0,1}, \vartheta^{0,1}_1\> + \<\beta, \vartheta^{1,0}_2\> \\
    \sfQ_2^+ & = & \<\sfp^{1,0}, \vartheta^{1,0}_2\> + \<\Bar{\beta}, \vartheta_1^{0,1}\> & , & \sfQ^-_2 & = & \<\sfp^{0,1}, \vartheta^{0,1}_2\> + \<\beta, \vartheta_1^{1,0}\> .
    \end{array}
  \]
  and the even function $\sfH$ as in \cref{prop:1dN=2}.
  These satisfy the Poisson brackets
  \begin{align*}
    \{\sfQ^+_\alpha, \sfQ^-_\beta\} & = \frac{(-1)^{\alpha+1}}{2} \delta_{\alpha\beta} \sfH , \qquad \mbox{for} \qquad {\alpha, \beta=1,2}\\
    \{\sfQ^\pm_\alpha, \sfQ^\pm_\beta\} & = 0.
  \end{align*}
  In other words, the functions $\sfQ^\pm_\alpha,\sfH$ determine a Hamiltonian action on $\sfX$ by the 1d $\cN=4$ supertranslation algebra.
\end{prop}

The proof is similar to that of \cref{prop:1dN=2}, so we omit it.

\subsection{Quantization}

Let us return to the phase space of supersymmetric mechanics $(\sfX, \omega)$ which is defined for any Riemannian manifold $(M, g)$.

\begin{lm}
  The symplectic supermanifold $(\sfX, \omega)$ is exact with primitive one-form
  \begin{equation}\label{eqn:primitive}
    \lambda = \sfp + \im \, a + g(\vartheta_1, \nabla \vartheta_1) - g(\vartheta_2, \nabla \vartheta_2) .
  \end{equation}
  In fact, there is a symplectomorphism $\sfX\cong \T^*(\Pi\T M)$, where the odd coordinates on the base are given by the components of $\vartheta_1 + \vartheta_2$.
\end{lm}

We will now apply the procedure of geometric quantization to the symplectic supermanifold $\sfX$. Choose a flat real line bundle $(L, \nabla_L, g)$ on $M$ equipped with a metric parallel with respect to $\nabla_L$. Equivalently, $(L, \nabla_L, g)$ can be encoded in the principal $\Z/2$-bundle of its unit frames. Define a connection on $\pi^* L$, a line bundle on $\sfX$, by the formula
\[
  \Tilde{\nabla} = \nabla_L + \lambda
\]
where $\lambda$ is the primitive \eqref{eqn:primitive} for the symplectic form $\omega$ on $\sfX$.
Notice that since $\nabla_L$ is flat, the curvature of $\Tilde{\nabla}_L$ is automatically $\omega$. Then $(\pi^* L, \Tilde{\nabla}, g)$ defines a prequantization of $\sfX$. We also have a polarization on $\sfX$ given by the fibers of $\sfX\rightarrow \Pi\T M$.

\begin{lm}
The geometric quantization of $\sfX$ is the $\Z/2$-graded vector space
\[\cH = \Gamma(\Pi\T M, \pi^* L) \cong \Gamma(M, \wedge^\bullet \T^*_M\otimes L).\]
\end{lm}

The supersymmetry action on the phase space quantizes to an action of the 1d $\cN=2$ supertranslation algebra on $\cH$ with
\begin{equation}
\sfQ_1 + \sfQ_2 = \nabla + (\alpha+\im a)\wedge(-).
\label{eq:supersymmetrytwisteddeRham}
\end{equation}
In addition, it admits an $R$-symmetry group $G_R=\SO(1, 1) = \R^\times$ acting as the grading operator on $\wedge^\bullet \T^*_M$. From \eqref{eq:supersymmetrytwisteddeRham} we obtain the following.

\begin{prop}
The A twist of $\cH$ is given by the twisted de Rham complex
\[(\Gamma(M, \Sym(\T^*_M[-1])\otimes L), \nabla + (\alpha+\im a)\wedge(-)).\]
\label{prop:Atwist}
\end{prop}

Let us now assume that $M$ has a K\"ahler structure and the $(1, 0)$ parts of $\alpha$ and $a$ are closed. 

By \cref{prop:1dN=4} the phase space $\sfX$ carries an action of the 1d $\cN=4$ supertranslation algebra which gives rise to an action on $\cH$ given by
\begin{align*}
\sfQ_1^+ + \sfQ_2^+ &= \partial + (\Bar{\beta}+\im \, a^{0, 1})\wedge(-) \\
\sfQ_1^- + \sfQ_2^- &= \Bar{\partial} + (\beta + \im \, a^{1, 0})\wedge(-) \\
\end{align*}

We are now ready to state the main observation of this paper.

\begin{thm}
Consider the family of twisting supercharges $\sfQ_\hbar = \sfQ_1^- + \sfQ_2^- + \hbar(\sfQ_1^+ + \sfQ_2^+)$ parametrized by $\hbar\in\CP^1$. Then:
\begin{itemize}
\item The B twist ($\hbar = 0$) of $\cH$ is \[\left(\Omega^\bullet(M; L)\; , \; \Bar{\partial} + (\beta + \im \, a^{1, 0})\wedge(-)\right).\]
\item The A twist ($\hbar = 1$) of $\cH$ is \[\left(\Omega^\bullet(M; L) \; , \; \nabla + (\alpha + \im \, a)\wedge(-)\right).\]

\item Let \[A=\Gamma(M, \Sym(\Omega^{0, 1}_M[-1])\otimes \Sym(\T^{1, 0}_M[1]))\] be the dg algebra equipped with the differential $\dbar + \beta \wedge (-)$.
The family $\sfQ_\hbar$ provides a BV quantization of the $A$-module $\cH$, where the Poisson bracket on $A$ is the Schouten bracket.
\end{itemize}
\label{thm:maintheorem}
\end{thm}
\begin{proof}
The first two claims are clear. Let us now prove the last claim. The family of twisting supercharges is
\[\sfQ_\hbar = \Bar{\partial} + (\beta+\im \, a^{1, 0})\wedge(-) + \hbar(\partial + (\Bar{\beta} + \im \, a^{0, 1})\wedge(-)).\]
As explained in \cref{ex:realdifferentialformsBV,ex:complexdifferentialformsBV}, $\partial$ viewed as a differential operator on the $A$-module $\cH$ has order 2 with symbol the Schouten bracket. At the same time $(\Bar{\beta}+\im a^{0, 1})\wedge(-)$ has order 1, so it does not contribute to the symbol.
\end{proof}

Consider the derived zero locus $\sfX = \R\beta^{-1}(0)$ of the one-form $\beta$. Its derived algebra of functions is
\[A=\R\Gamma(\sfX, \cO)\cong\Gamma(M, \Sym(\Omega^{0, 1}_M[-1])\otimes \Sym(\T^{1, 0}_M[1]))\]
equipped with the differential $\Bar{\partial} + \beta\wedge(-)$. As explained in \cref{sect:BVquantization}, the algebra $A$ carries a degree $1$ Poisson structure given by the Schouten bracket. Define the local system $\cL_a = (\cO_{\sfX}, \d + \im a^{1, 0})$ on $\sfX$. The space $\sfX$ admits a canonical square root $\det(\bL_{\sfX})^{\frac12}$ of the canonical bundle given by the canonical bundle of $M$. Consider the dg $A$-module
\[M = \R\Gamma(\sfX, \det(\bL_{\sfX})^{\frac12}\otimes L\otimes \cL_a)[-\dim M]\cong \Omega^\bullet(M; L)\]
equipped with the differential $\Bar{\partial} + (\beta + \im a^{1, 0})\wedge(-))$. The main claim of \cref{thm:maintheorem} is that the deformation from the B twist to the A twist corresponds to a BV quantization of the former.

The above observation allows us to give a model of the A twist even when the base manifold $M$ is infinite-dimensional (for instance, an infinite-dimensional Fr\'echet manifold). Indeed, according to \cref{conj:perverseBV} we may model a BV quantization of $\sfX$ via the cohomology of the perverse sheaf $P_\sfX$ on the zero locus $\beta^{-1}(0) = t_0(\sfX)$. The choice of the orientation data is manifested in the finite-dimensional situation in the freedom of choosing $L$.

\begin{proposal}
Let $M$ be an (infinite-dimensional) K\"ahler manifold and $\alpha,a$ one-forms whose $(1, 0)$ parts $\beta=\alpha^{1, 0}, a^{1, 0}$ are closed. Suppose the zero locus $\beta^{-1}(0)=X$ admits the structure of a $(-1)$-shifted symplectic algebraic scheme $\sfX$ (so that $t_0(\sfX) = X$) equipped with an orientation data $\det(\bL_\sfX)^{\frac12}$. Consider the supersymmetric quantum mechanics into $M$ twisted by $\alpha+\im a$. Then the space of states in the A twist, for generic parameters, is
\[\cH = \R\Gamma(X, P_{\sfX}\otimes \cL_a).\]
\label{mainproposal}
\end{proposal}

\begin{remark}
Let us explain the caveat regarding generic parameters in the above proposal. By \cref{thm:maintheorem} the twist with respect to $\sfQ_\hbar$ considered over $k[\hbar]$ is a BV quantization. By \cref{conj:perverseBV} if we work over $k(\!(\hbar)\!)$ the BV quantization is isomorphic to the cohomology of the perverse sheaf $P_\sfX$. But if the BV quantization is defined over $k[\hbar]$, the vector space for generic $\hbar$ (i.e. over $k(\hbar)$) will have the same dimension as the $k(\!(\hbar)\!)$-vector space.
\end{remark}

\begin{remark}
We expect that the definition of the perverse sheaf $P_\sfX$ may be extended to the setting of derived complex-analytic geometry, so the choice of the algebraic structure on $\sfX$ will not matter.
\end{remark}

\section{Gauged supersymmetric mechanics}
\label{sect:gaugedmechanics}

In this section we recall the coupling of supersymmetric mechanics to gauge theory. We also write down the action functionals of the corresponding theories.

\subsection{$\cN=2$ case}

As before, let $S$ be a one-dimensional oriented Riemannian manifold. Let $\d t$ be a positive normalized frame of the cotangent bundle. Consider the following additional data:
\begin{itemize}
\item $G$ is a Lie group equipped with a nondegenerate symmetric bilinear pairing on its Lie algebra $\fg$.

\item $M$ is a Riemannian manifold equipped with a $G$-action by isometries.

\item $\alpha$ and $a$ are one-forms on $M$, which are $G$-invariant and equivariantly closed: $\d\alpha = 0$ and $\iota_{\xi_x} \alpha = 0$ for every $x\in\g$, where $\xi_x$ is the vector field given by the infinitesimal $G$-action.

\item $h$ is a locally constant $G$-invariant function on $M$.
\end{itemize}

The $G$-gauged $\cN=2$ supersymmetric mechanics has the following fields:

\begin{itemize}
\item A principal $G$-bundle $P\rightarrow S$.

\item A connection $A$ on $P$.

\item An odd section $\eta\in \Pi\Gamma(S, \ad P)$.

\item An odd one-form $\lambda \in \Pi\Omega^1(S, \ad P)$. 

\item Sections $\varphi, \xi \in\Gamma(S, \ad P)$.

\item A section $\phi$ of $P\times^G M\rightarrow S$.

\item An odd section $\chi$ of $\phi^*(P\times^G \T M) \rightarrow S$.

\item An odd section $\psi$ of $\phi^*(P\times^G \T M)\otimes \T^* S\rightarrow S$.
\end{itemize}

\begin{remark}
Here the fields $(A, \eta, \lambda, \xi, \varphi)$ belong to the gauge multiplet and $(\phi, \chi, \psi)$ to the matter multiplet. 
\end{remark}

If $\sigma$ is a section of $\ad P \otimes \cA$, where $\cA$ is some bundle on $S$, then we denote by $\Hat{\sigma}$ the induced $\cA$-valued vector field on $P \times^G M$. For instance, when $M$ is a vector space equipped with a linear orthogonal $G$-action, then $\Hat{\sigma} = \sigma \cdot \phi$.

The variation of the fields with respect to the A supercharge $Q_A$ is given by
\begin{equation}
\begin{cases}
\delta \phi^i = \im \, \chi^i \\
\delta \chi = \Hat{\varphi} \\
\delta \xi = \im \, \eta \\
\delta \eta = [\varphi,\xi] \\
\delta \psi^i = -\d_A\phi^i - \im \, \Gamma^i_{jk} \chi^j \psi^k - \alpha^i\d t \\
\delta A = \im \, \lambda \\
\delta \lambda = - \d_A \varphi \\
\delta \varphi = 0 
\end{cases}
\label{eq:1dAgaugedN=2supercharge}
\end{equation}

The bosonic part of the action is given by (see e.g. \cite[Theorems 3.46 and 6.33]{DeligneFreed})
\begin{equation}
S_{bosonic} = \int_S \dvol_S\left(\frac{1}{2}|\d_A\phi|^2 + \frac{1}{2}|\alpha|^2 + (\d_A\varphi, \d_A \xi) - \frac{1}{2}|[\varphi, \xi]|^2 + (\hat{\varphi},\hat{\xi}) + \phi^* h\right) + \im \int_S \phi^* a
\label{eq:1dN=2gaugedsigmamodel}
\end{equation}

\begin{remark}
Let us explain the meaning of the term $\int_S \dvol_S\, \phi^* h$ in the action. The function $h$ gives a grading operator $P_h$ on the Hilbert space of the theory. Let $\widehat{H}$ be the Hamiltonian operator. Then the addition of the term $\int_S \dvol_S\, \phi^* h$ to the action corresponds to deforming the partition function $\tr(\exp(-\widehat{H}))$ to $\tr(\exp(-\widehat{H})\exp(-P_h))$.
\end{remark}

\subsection{$\cN=4$ case}

As in \cref{sect:Kahlermechanics}, in the case the target manifold is K\"ahler, the supersymmetry is enhanced. Consider the following data:
\begin{itemize}
\item $G$ is a Lie group equipped with a nondegenerate symmetric bilinear pairing on its Lie algebra.

\item $M$ is a K\"ahler manifold equipped with a $G$-action by K\"ahler isometries with the moment map $\mu\colon M\rightarrow \g^*$, so that $\d \mu(x) = -\iota_{\xi_x} \omega$ for every $x\in\g$.

\item $\alpha,a$ are real one-forms on $M$ whose $(1, 0)$ parts are $G$-invariant and equivariantly closed.

\item $h$ is a locally constant $G$-invariant function on $M$.
\end{itemize}

The $G$-gauged $\cN=4$ supersymmetric mechanics $S\rightarrow M$ has the following fields:
\begin{itemize}
\item A principal $G$-bundle $P\rightarrow S$.

\item A connection $A$ on $P$.

\item Odd sections $\eta, c, \nu\in\Pi\Gamma(S, \ad P)$.

\item An odd one-form $\lambda\in\Pi\Omega^1(S, \ad P)$.

\item Sections $\varphi,\xi, \sigma\in\Gamma(S, \ad P)$.

\item A section $\phi$ of $P\times^G M\rightarrow S$.

\item An odd section $\chi$ of $\phi^*(P\times^G \T M) \rightarrow S$.

\item An odd section $\psi$ of $\phi^*(P\times^G \T M)\otimes \T^* S\rightarrow S$.
\end{itemize}

The variation of the fields with respect to the A supercharge $Q_A$ is given by
\begin{equation}
\begin{cases}
\delta \phi^i = \im \, \chi^i \\
\delta \chi = \Hat{\varphi} \\
\delta \xi = \im \, \eta \\
\delta \eta = [\varphi, \xi] \\
\delta \psi^i = -\d_A \phi^i - \im\Gamma^i_{jk} \chi^j \psi^k - (\alpha^i + \sigma^j \partial_i \mu(e_j))\d t \\
\delta A = \im \, \lambda \\
\delta \sigma = \im \, \nu \\
\delta \lambda = -\d_A\varphi \\
\delta \nu = [\varphi, \sigma] \\
\delta c = -\star\d_A\sigma - \mu \\
\delta \varphi = 0.
\end{cases}
\label{eq:1dAgaugedN=4supercharge}
\end{equation}

The bosonic part of the action is (see \cite[Section 3]{Baptista} for the case $\alpha = 0$)
\begin{multline}
S_{bosonic} = \int_S \dvol_S\bigg(\frac{1}{2}|\d_A\sigma|^2 + \frac12 |\Hat{\sigma}|^2 + \frac{1}{2}|\d_A\phi|^2 + \frac{1}{2}|\mu|^2 + \frac{1}{2}|\alpha|^2 + (\d_A\varphi, \d_A \xi) \\
 - \frac{1}{2}|[\varphi, \xi]|^2 + (\hat{\varphi},\hat{\xi}) + ([\varphi,\sigma], [\xi,\sigma]) + \phi^* h\bigg) + \im \int_S \phi^* a 
\label{eq:1dN=4gaugedsigmamodel}
\end{multline}

\begin{remark}
Given the data as above, consider the Riemannian manifold $\tilde{M} = M\times \g$. Given a basis $\{e_i\}$ of $\g$ we denote the corresponding coordinates on the second factor by $\sigma^i$. Then the action \eqref{eq:1dN=4gaugedsigmamodel} of $\cN=4$ gauged supersymmetric mechanics coincides with the action \eqref{eq:1dN=2gaugedsigmamodel} of the $\cN=2$ gauged supersymmetric mechanics into $\tilde{M}$ with the one-form
\[\tilde{\alpha} = \alpha + \sum_i (\mu(e_i) \d\sigma^i + \sigma^i \d\mu(e_i))\]
if we match the fields as in \cref{tbl:N=2toN=4}.

\begin{table}[h]
\begin{tabular}{cc}
\hline
$\cN=2$ target $M \times \fg$ & $\cN=4$ target $M$ \\
\hline
$\phi$ & $\phi,\sigma$ \\
$\chi$ & $\chi,\nu$ \\
$\psi$ & $\psi, c$ \\
$A, \eta, \lambda, \varphi , \xi$ & $A, \eta, \lambda, \varphi , \xi$\\
\hline
\end{tabular}
\caption{1d $\cN=4$ fields in a $\cN=2$ description}
\label{tbl:N=2toN=4}
\end{table}
\label{rmk:N=2toN=4}
\end{remark}

Suppose $G$ acts freely and properly on $\mu^{-1}(0)\subset M$, so that $M/\!/G = \mu^{-1}(0) / G$ is a K\"ahler manifold equipped with a closed $(1, 0)$ form $[\beta]$. It is shown in \cite[Section 6B]{HKLR} that the low-energy approximation to the $G$-gauged $\cN=4$ supersymmetric mechanics $S\rightarrow M$ described by the action \eqref{eq:1dN=4gaugedsigmamodel} is described by the supersymmetric mechanics $S\rightarrow M/\!/G$.

Let $G_\C\supset G$ be a complex Lie group whose Lie algebra is the complexification of the Lie algebra of $G$. Moreover, suppose the $G$-action on $M$ extends to a holomorphic $G_\C$-action on $M$. We may then identify the K\"ahler quotient $M/\!/G$ with the GIT quotient $(M/G_\C)_s$, where the stable locus $(M/G_\C)_s\subset M/G_\C$ consists of $G_\C$-orbits intersecting $\mu^{-1}(0)$. Using the previous two observations we introduce the following version of \cref{mainproposal} in the presence of gauge symmetries.

\begin{proposal}
Let $G$ be a Lie group with complexification $G_\C$, $M$ be a K\"ahler manifold equipped with a $G$-structure preserving the K\"ahler structure with a moment map $\mu\colon M\rightarrow \g^*$, $\alpha$ and $a$ one-forms whose $(1, 0)$ parts $\beta=\alpha^{1, 0}$ and $a^{1, 0}$ are equivariantly closed and. Consider the induced one-form $[\beta]$ on the quotient stack $M/G_\C$. Consider a locally constant function $h$ on $M/G_\C$. Suppose the zero locus $[\beta]^{-1}(0)=X$ admits the structure of a $(-1)$-shifted symplectic algebraic stack $\sfX$ equipped with an orientation data $\det(\bL_\sfX)^{\frac12}$. Consider the $G$-gauged $\cN=4$ supersymmetric quantum mechanics into $M$ twisted by $\alpha+\im a$. Then the space of states in the A twist, for generic parameters, is
\[\cH = \R\Gamma(X, P_{\sfX}\otimes \cL_a).\]
The locally constant function $h$ on $X$ defines a grading on $\cH$.
\label{mainproposalstack}
\end{proposal}

\section{Compactification of principal bundles}
\label{sect:bundlecompactification}

In this short section we collect some results that will be useful for describing compactifications of gauge theories.

\subsection{Principal bundles}

Let $G$ be a finite-dimensional Lie group, $S, N$ closed manifolds and consider a principal $G$-bundle $P\rightarrow S\times N$.

\begin{defn}
A principal $G$-bundle $P\rightarrow S\times N$ is \defterm{trivializable along the fibers of $S\times N\rightarrow S$} if it admits a trivializing cover $\{U_i\times N\}$, where $\{U_i\}$ is an open cover of $S$.
\label{def:trivializablebundle}
\end{defn}

The space of smooth maps $\Map(N, G)$ forms a Fr\'echet Lie group under pointwise multiplication. Its Lie algebra is $C^\infty(N; \g)$. One has the following description of principal $\Map(N, G)$-bundles.

\begin{prop}
There is a 1:1 correspondence between isomorphism classes of principal $G$-bundles $P\rightarrow S\times N$ trivializable along the fibers of $\pi\colon S \times N\rightarrow S$ and isomorphism classes of principal $\Map(N, G)$-bundles $P_N\rightarrow S$.
\label{prop:bundlecorrespondence}
\end{prop}
\begin{proof}
Given a principal $\Map(N, G)$-bundle $P_N\rightarrow S$ consider its pullback $\pi^* P_N\rightarrow S\times N$. There is a natural evaluation map $\Map(N, G)\times N\rightarrow G$, so we can induce the principal $\Map(N, G)$-bundle $\pi^* P_N$ to a principal $G$-bundle. If $\{U_i\}$ is a trivializing cover for $P_N\rightarrow S$, then $\{U_i\times N\}$ is a trivializing cover for $P\rightarrow S\times N$.

Conversely, suppose $P \rightarrow S \times N$ is a principal $G$-bundle that is trivializable along the fibers of $\pi$. Let $\{U_i\times N\}$ be the trivializing cover. Then $\pi_*(P|_{U_i\times N})$ defines a $\Map(N, G)$-torsor on $U_i$. Gluing these defines a principal $\Map(N, G)$-bundle $P_N\rightarrow S$.
\end{proof}

Under the above correspondence we may identify the adjoint bundle $\ad P_N$ with the pushforward $\pi_* \ad P$ of the adjoint bundle on $S\times N$.

\subsection{Connections}

Let us now discuss connections on principal $\Map(N, G)$-bundles.

\begin{prop}
Under the correspondence given by \cref{prop:bundlecorrespondence} a connection on a principal $\Map(N, G)$-bundle $P_N\rightarrow S$ corresponds to a connection on $P\rightarrow S\times N$ in the $S$ direction.
\label{prop:connectionScorrespondence}
\end{prop}
\begin{proof}
Let $\{U_i\}$ be a trivializing cover for $P_N\rightarrow S$ and $g_{ij}\colon U_i\cap U_j\rightarrow \Map(N, G)$ the transition functions. Then a connection on $P_N$ is specified by a collection of one-forms $A_i\in\Omega^1(U_i; \Map(N, \g))$ satisfying $\Ad_{g_{ij}}(\d + A_j) = \d + A_i$. This data is obviously the same as a connection on $P$ in the $S$ direction.
\end{proof}

Let $\Conn^{\triv}_G(N)$ be the affine space of connections on the trivial $G$-bundle on $N$. It carries an action of $\Map(N, G)$ given by gauge transformations.

\begin{prop}
Under the correspondence given by \cref{prop:bundlecorrespondence} a $\Map(N, G)$-equivariant map $\phi\colon P_N\rightarrow \Conn^{\triv}_G(N)$, i.e. a section of $P_N\times^{\Map(N, G)} \Conn^{\triv}_G(N)\rightarrow S$, is the same as a connection on $P\rightarrow S\times N$ in the $N$ direction.
\label{prop:connectionNcorrespondence}
\end{prop}
\begin{proof}
Let $\{U_i\}$ be a trivializing cover for $P_N\rightarrow S$ and $g_{ij}\colon U_i\cap U_j\rightarrow \Map(N, G)$ the transition functions. A $\Map(N, G)$-equivariant map $\phi\colon P_N\rightarrow \Conn^{\triv}_G(N)$ is the same as a collection of maps $\phi_i\colon U_i\rightarrow \Conn^{\triv}_G(N)$ such that $g_{ij} \phi_j = \phi_i$. A map $\phi_i\colon U_i\rightarrow \Conn^{\triv}_G(N)$ is the same as a $\g$-valued one-form along the fibers of $U_i\times N\rightarrow U_i$. This data is the same as a connection on $P\rightarrow S\times N$ in the $N$ direction.
\end{proof}

The following statement is proven analogously.

\begin{prop}
Suppose $A_S$ is a connection on a principal $\Map(N, G)$-bundle $P_N\rightarrow S$ and $\phi$ a section of $P_N\times^{\Map(N, G)} \Conn^{\triv}_G(N)$ corresponding to a connection $A_N$ on $P\rightarrow S\times N$ in the $N$ direction by \cref{prop:connectionNcorrespondence}. Then $\d_{A_S} \phi$ coincides with the $\Gamma(N\times S, \Omega^1_N\otimes \Omega^1_S\otimes \ad P)$ component of the curvature of the connection $A_N + A_S$ on $P\rightarrow S\times N$.
\end{prop}

\section{2d A-model}
\label{sect:2dA}

In this section consider the 2d A-model into a hyperK\"ahler manifold on the interval with supersymmetric boundary conditions and its compactification to supersymmetric mechanics.

\subsection{2d $\sigma$-model}

Let $(M, g, \omega, I)$ be a K\"ahler manifold and $(\Sigma, h, j)$ a Riemann surface. We say a one-form $\psi\in\Omega^1(\Sigma, \phi^* \T_M)$ is \defterm{self-dual} if
\[(j\otimes I)\psi = \psi.\]
A self-dual two-form has components $\psi = \psi_+ + \psi_-$, where
\[\psi_+\in\Omega^{1, 0}(\Sigma, \phi^* \T^{(0, 1)}_M),\qquad \psi_-\in\Omega^{0, 1}(\Sigma, \phi^* \T^{(1, 0)}_M).\]

The 2d A-model into $M$ has the following fields \cite{WittenSigma}:
\begin{itemize}
\item A map $\phi\colon \Sigma\rightarrow M$.
\item An odd section $\chi\in\Pi\Gamma(\Sigma, \phi^* \T_M)$.
\item A self-dual odd one-form $\psi\in\Pi\Omega^1(\Sigma, \phi^* \T_M)$.
\end{itemize}

The supersymmetry transformation is
\begin{equation}
\begin{cases}
\delta \phi^i = \im\chi^i \\
\delta \chi^i = 0 \\
\delta \psi^{\ibar} = -\partial \phi^{\ibar} - \im \chi^{\jbar}\Gamma^{\ibar}_{\jbar \mbar}\psi^{\mbar} \\
\delta \psi^i = -\Bar{\partial}\phi^i -\im \chi^j \Gamma^i_{jm}\psi^m.
\end{cases}
\label{eq:2dSUSY}
\end{equation}

The bosonic part of the action is
\[S_{bosonic} = \frac{1}{2}\int_\Sigma \dvol_\Sigma (\d\phi, \d\phi).\]

\subsection{Supersymmetric mechanics on the path space}

Let $L_0, L_1\subset M$ be two Lagrangian submanifolds. Consider the path space
\[\rP(L_0, L_1) = \{\phi\colon [0, 1]\rightarrow M\ |\ \phi(0)\in L_0,\ \phi(1)\in L_1\},\]
where the coordinate along $[0, 1]$ will be denoted by $s$. It has a natural structure of a Fr\'echet manifold (see e.g. \cite{Stacey}) such that the tangent space at $\phi\in\rP(L_0, L_1)$ can be identified with the space
\[\T_\phi \rP(L_0, L_1) = \{v\in\Gamma([0, 1], \phi^* \T_M)\ |\ v(0)\in \T_{\phi(0)} L_0,\ v(1)\in \T_{\phi(1)} L_1\}.\]
It has a natural (weak) Riemannian metric defined by
\[(v, w) = \int_0^1 (v(s), w(s))\d s.\]
We will also be interested in the closed submanifold $\widetilde{\rP}(L_0, L_1)\subset \rP(L_0, L_1)$ defined as
\[\widetilde{\rP}(L_0, L_1) = \{\phi\in \rP(L_0, L_1)\ |\ \phi'(0)\perp L_0,\ \phi'(1)\perp L_1\},\]
where the orthogonality is defined with respect to the metric on $M$. We have $\T_\phi \widetilde{\rP}(L_0, L_1)\subset \T_\phi\rP(L_0, L_1)$ defined by the conditions $v(0)=v(1) = 0$.

The path space $\rP(L_0, L_1)$ carries a natural one-form $\alpha\in\Omega^1(\rP(L_0, L_1))$ defined by
\begin{equation}
\iota_v \alpha = \int_0^1 \omega(v(s), \phi'(s))\d s.
\label{eq:symplecticaction}
\end{equation}
It is easy to see that it is closed. Its primitive, whenever defined, is known as the \defterm{symplectic action functional}.

Let $S$ be an oriented Riemannian 1-manifold. Let $\Sigma = S\times [0, 1]$. Consider the 2d A-model on $\Sigma$ with the boundary conditions
\[
\begin{cases}
\phi(s, 0) \in L_0 \\
\frac{\partial\phi}{\partial s}(s, 0) \perp L_0 \\
\chi(s, 0)\in \phi^* \T_{L_0} \\
\psi|_{S\times \{0\}} \in \Omega^1(S, \phi^* \T_{L_0}).
\end{cases}
\]
Similar boundary conditions are imposed at the other end of the interval $[0, 1]$.

\begin{thm}
The 2d A-model of maps $S\times [0, 1]\rightarrow M$ with the above boundary conditions is equivalent to the supersymmetric mechanics of maps $S\rightarrow \rP(L_0, L_1)$, where $\rP(L_0, L_1)$ is equipped with its natural Riemannian structure and one-form $\alpha$ given by \eqref{eq:symplecticaction}. Under this correspondence the supersymmetry transformation \eqref{eq:2dSUSY} corresponds to the transformation induced by the A supercharge $Q_A$ \eqref{eq:1dAgaugedN=2supercharge}.
\end{thm}
\begin{proof}
A map $\phi\colon S\times [0, 1]\rightarrow M$ satisfying the above boundary conditions is the same as a map $\phi\colon S\rightarrow \rP(L_0, L_1)$. A self-dual one-form $\psi\in\Omega^1(S\times [0, 1], \phi^* \T_M)$ can be written as $\psi = \psi_t \d t + I \psi_t \d s$. Therefore, we can match the fields of the 2d A-model on $S\times [0, 1]$ and the supersymmetric mechanics as shown in \cref{tbl:2dAcompactification}.

We have
\begin{align}
\iota_v \alpha &= \int_0^1 \omega(v(s), \phi'(s))\d s \nonumber \\
&= -\int_0^1 (v(s), I \phi'(s)) \d s \nonumber \\
&= -(v, I\partial_s\phi). \label{eq:alphadual}
\end{align}
Therefore,
\[|\alpha|^2 = \int_0^1 |I\phi'(s)|^2 \d s = \int_0^1 |\phi'(s)|^2 \d s\]
since $I$ is orthogonal. In particular,
\begin{align*}
S_{bosonic} &= \frac{1}{2}\int_\Sigma \dvol_\Sigma |\d\phi|^2 \\
&= \frac{1}{2}\int_S\dvol_S(|\partial_t \phi|^2 + |\alpha|^2)
\end{align*}
which coincides with the bosonic part of the action of the supersymmetric mechanics.

The supersymmetry transformation induced by $Q_A$ is
\[
\begin{cases}
\delta \phi^i = \im\,\chi^i \\
\delta \psi^{\ibar} = -\partial_t \phi^{\ibar} - \im \,\Gamma^{\ibar}_{\jbar \kbar} \chi^{\jbar}\psi^{\kbar} - \alpha^{\ibar} \\
\delta \psi^i = -\partial_t \phi^i - \im \,\Gamma^i_{jk} \chi^j\psi^k - \alpha^i \\
\delta \chi^i = 0.
\end{cases}
\]

Using \eqref{eq:alphadual} we get $\alpha^{\ibar} = \im \, \partial_s \phi^{\ibar}$ and $\alpha^i = -\im \,\partial_s \phi^i$. Observing that
\[\partial \phi^{\ibar} = \partial_t \phi^{\ibar} + \im \, \partial_s \phi^{\ibar},\qquad \Bar{\partial} \phi^i = \partial_t \phi^i - \im \, \partial_s \phi^i\]
we see that the above supersymmetry transformation coincides with \eqref{eq:2dSUSY}.
\end{proof}

\begin{table}[h]
\begin{tabular}{cc}
\hline
supersymmetric mechanics & 2d A-model \\
\hline
$\phi$ & $\phi$ \\
$\chi$ & $\chi$ \\
$\psi$ & $\psi_t$ \\
\hline
\end{tabular}
\caption{Fields in the supersymmetric mechanics and in the 2d A-model.}
\label{tbl:2dAcompactification}
\end{table}

\subsection{HyperK\"ahler case}

Let us now assume that $M$ has a hyperK\"ahler structure. We denote by $I, J, K$ the three complex structures, $\omega_I, \omega_J, \omega_K$ the three K\"ahler structures and $\Omega_I, \Omega_J, \Omega_K$ the three holomorphic symplectic structures. Suppose that $L_0, L_1\subset M$ are holomorphic Lagrangians with respect to the $\Omega_I$ holomorphic symplectic structure (these give examples of $(B, A, A)$ branes).

As before, $\widetilde{P}(L_0, L_1)$ has an induced (weak) Riemannian structure. The complex structure $I$ induces pointwise a K\"ahler structure on $\widetilde{P}(L_0, L_1)$. We may also consider the closed one-form $\beta\in\Omega^{1, 0}(\widetilde{P}(L_0, L_1))$ defined by
\[\beta(v) = \int_0^1 \Omega_I(v(s), \phi'(s)) \d s.\]
By construction its real part is
\[\alpha(v) = \int_0^1 \omega_J(v(s), \phi'(s)) \d s.\]

Therefore, by \cref{prop:1dN=4} the supersymmetry of the supersymmetric mechanics of maps $S\rightarrow \widetilde{P}(L_0, L_1)$ enhances from $\cN=2$ to $\cN=4$. Let us now assume that $(M, I)$ has an algebraic structure and $L_0, L_1\subset M$ are algebraic subvarieties. The zero locus of $\beta$ is the intersection $L_0\cap L_1$. It has a natural enhancement to a $(-1)$-shifted symplectic scheme $L_0\times_M L_1$ \cite[Theorem 2.9]{PTVV}. The \cref{mainproposal} suggests the following.

\begin{proposal}
Suppose $M$ is a hyperK\"ahler manifold, so that $(M, I, \Omega_I)$ is an algebraic symplectic manifold. Suppose $L_0, L_1\subset M$ are algebraic Lagrangian subvarieties. Choose square roots of the canonical bundles on $L_i$ which determine an orientation data on the derived intersection $L_0\times_M L_1$. Then the space of states $\Hom(L_0, L_1)$ in the 2d A-model into $(M, \omega_J)$, for generic parameters, is
\[\R\Gamma(L_0\cap L_1, P_{L_0\times_M L_1}).\]
\end{proposal}

\begin{remark}
The above definition was proposed in \cite{BehrendFantechi,BBDJS}. We refer to \cite{SolomonVerbitsky} for the discussion of the relationship to the usual definition of the Floer homology group. Note also that the above definition is independent of the choice of the algebraic structure if we use the perverse sheaf constructed in \cite{Bussi}.
\end{remark}

\begin{remark}
One may interpret ``generic parameters'' in the above statement as follows. Let us consider the 2d A-model into $M$ with the symplectic structure $\omega = (1+\alpha) \omega_J$ and the $B$-field $B = \alpha \omega_K$ for $\alpha$ a non negative number. Then for $\alpha=0$ we recover the usual 2d A-model into $(M, \omega_J)$. In the limit $\alpha\rightarrow \infty$ we obtain the 2d B-model into $(M, I)$ \cite{KapustinNC}. So, one expects the proposal to be true for large $\alpha$.
\end{remark}

\section{3d A-model}
\label{sect:3dA}

In this section we describe a 3-dimensional analog of the A-model with target given by a hyperK\"ahler manifold.

\subsection{3d $\sigma$-model}

Fix the following data:
\begin{itemize}
\item A Lie group $G$ equipped with a homomorphism $\rho\colon G\rightarrow \SO(3)$.

\item $X$ a hyperK\"ahler manifold equipped with a $G$-action by isometries. In addition, we assume it acts on the sphere of complex structures on $X$ via $\rho$.

\item $M$ is an oriented Riemannian 3-manifold equipped with a principal $G$-bundle $Q$ together with a connection $\nabla$ and an isometric identification
\begin{equation}
Q\times^G \RR^3\cong \Omega^1_M.
\label{eq:Qisometry}
\end{equation}
\end{itemize}

Consider the induced bundle
\[\fX = Q \times^G X\xrightarrow{\pi} M.\]
The connection $\nabla$ on $Q$ induces a connection on $\fX\rightarrow M$ that we denote by the same letter. Its vertical tangent bundle is
\[V \fX = Q \times^G \T X.\]
Since $G$ acts on $X$ by isometries, $\fX\rightarrow M$ admits a fiberwise metric. In addition, since $G$ acts by permuting the complex structures on $X$, the bundle $\fX\rightarrow M$ carries a fiberwise hyperK\"ahler structure with the associated sphere bundle of complex structures given by the unit sphere bundle of $M$ using \eqref{eq:Qisometry}. In particular, the action by complex structures gives a map
\[a\colon \pi^* \Omega^1_M\otimes_\R V \fX\longrightarrow V \fX.\]

The 3d $A$-model has the following fields:
\begin{itemize}
\item A section $\phi$ of $\fX\rightarrow M$.
\item A pair of odd sections $\chi,\psi\in\Pi\Gamma(M, \phi^* V \fX)$.
\end{itemize}

The bosonic part of the action functional of the 3d A-model is
\[
S_{bosonic} = \frac12 \int_M \dvol_M \; (\nabla \phi, \nabla \phi)
\]

\begin{remark}
This theory may be obtained by topologically twisting the 3-dimensional supersymmetric $\sigma$-model into $X$. The corresponding 4-dimensional version was considered in \cite{AnselmiFre,FKS}.
\end{remark}

The covariant derivative gives a section $\nabla\phi\in\Gamma(M, \Omega^1_M\otimes \phi^* V\fX)$. Applying the action map $a$ by complex structures, we get $a(\nabla \phi)\in\Gamma(M, \phi^* V\fX)$.

\begin{remark}
In local coordinates $(x, y, z)$ on $M$ we have
\[a(\nabla \phi) = I \nabla_x \phi + J \nabla_y \phi + K \nabla_z \phi.\]
The equation $a\circ \nabla\phi=0$ is known as the 3-dimensional Fueter equation, see e.g. \cite{Walpuski}.
\end{remark}

The supersymmetry transformation is 
\begin{equation}
\begin{cases}
\delta \phi^i = \im\chi^i \\
\delta \chi^i = 0 \\
\delta \psi^i = (I a(\nabla \phi))^i - \im \, \Gamma^i_{jk} \chi^j \psi^k .
\end{cases}
\label{eq:3dSUSY}
\end{equation}

\subsection{Twisted hyperK\"ahler mapping space}

Suppose the homomorphism $\rho\colon G\rightarrow \SO(3)$ factors through $G\rightarrow \SO(2)\hookrightarrow \SO(3)$. Our convention is that $G$ fixes the complex structures $\pm I$ on $X$ and acts on $J, K$ via~$\rho$.

Let $\Sigma$ be a Riemann surface with a complex structure we denote by $j$. In addition, suppose $P\rightarrow \Sigma$ is a principal $G$-bundle with a connection $\nabla$ and an isometric identification
\[P\times^G \R^2\cong \Omega^1_\Sigma.\]

Consider the space $X_\Sigma$ of smooth sections of the bundle
\[\fX_\Sigma = P\times^G X\xrightarrow{\pi_\Sigma} \Sigma.\]
As before, this is a bundle of hyperK\"ahler manifolds and there is an action
\[a_\Sigma\colon \pi_\Sigma^*\Omega^1_\Sigma\otimes_\R V \fX_\Sigma\longrightarrow V \fX_\Sigma\]
by complex structures. The tangent space at $\phi\in X_\Sigma$ may be identified with
\[\T_\phi X_\Sigma \cong \Gamma(\Sigma, \phi^* V \fX_\Sigma).\]

There is a (weak) Riemannian metric on $X_\Sigma$ defined by integrating the pointwise metric along $\Sigma$:
\[
(v, w)_{X_\Sigma} = \int_\Sigma \dvol_\Sigma \, (v , w),\qquad v,w\in\T_\phi X_\Sigma.
\]
The complex structure $I$ on $X$ induces a complex structure on $X_\Sigma$ in a similar way.

Define the one-form $\alpha$ on $X_\Sigma$ by 
\begin{equation}\label{eq:3dAalpha}
\iota_v \alpha = -\int_\Sigma \dvol_\Sigma \, (v, a(j \nabla \phi)),
\end{equation}
where $j\nabla s$ is obtained by acting by the complex structure $j$ on $\Sigma$ on the $\Omega^1_\Sigma$ factor of $\nabla \phi$.

If we choose a trivialization of $P$ over a coordinate neighborhood $(x,y)$ in $\Sigma$ this one-form becomes
\begin{align*}
\iota_v \alpha & = \int_\Sigma \big(-(v, J \nabla_y \phi) + (v, K \nabla_x \phi) \big) \d x \d y \\
&= \int_\Sigma (\omega_J(v, \nabla_y \phi) - \omega_K(v, \nabla_x \phi)\big) \d x \d y.
\end{align*}

The one-form $\alpha$ is the real part of the following holomorphic $(1,0)$ form $\beta$.  Let
\[\dbar \phi\in\Gamma(\Sigma, \Omega^{0,1}_\Sigma \otimes_{\RR} \phi^* V \fX_\Sigma)\]
be the $(0, 1)$ part of the covariant derivative. Define $\beta$ by the formula
\begin{equation}\label{eq:3dAbeta}
\iota_v \beta = 2\, \im \int_\Sigma \dvol_\Sigma \, (v, a(\dbar \phi)).
\end{equation}

Finally, define the function $h$ on $X_\Sigma$ by
\begin{equation}
h = \int_\Sigma \omega_I(\nabla_x \phi, \nabla_y \phi) \d x \d y.
\label{eq:3dAh}
\end{equation}

The function $h$ is locally constant and computes the symplectic volume of the section $\phi$ with respect to $\omega_I$.

\subsection{Supersymmetric mechanics on the mapping space}

Consider the setting as in the previous sections and take $M = S\times \Sigma$ for an oriented Riemannian 1-manifold $S$. Take $Q = S\times P$ with the trivial connection along the $S$ direction.

\begin{thm}
The 3d A-model of sections $\fX\rightarrow \Sigma\times S$ is equivalent to the supersymmetric mechanics of maps $S\rightarrow X_\Sigma$, where $X_\Sigma$ is equipped with its natural Riemannian structure, one-form $\alpha$ given by \eqref{eq:3dAalpha}, and $h$ as in \eqref{eq:3dAh}. Under this correspondence the supersymmetry transformation \eqref{eq:3dSUSY} corresponds to the transformation induced by the A supercharge $Q_A$ \eqref{eq:1dN=2gaugedsigmamodel}.
\end{thm}
\begin{proof}
A section $\phi$ of the bundle $\fX=S\times \fX_\Sigma\rightarrow S\times \Sigma$ is the same as a map $\phi\colon S\rightarrow X_\Sigma=\Sect(\Sigma, \fX_\Sigma)$. Using the identification
\[\Gamma(S, \phi^* \T_{X_\Sigma})\cong \Gamma(S\times M, \phi^* V \fX)\]
we can identify the fermion fields in the 3d A-model with the fermion fields in the supersymmetric mechanics.

Next, in local coordinates we have
\begin{align*}
\iota_v \alpha & = \int_\Sigma \big( (v, K \nabla_x \phi) - (v, J \nabla_y \phi) \big) \d x \d y \\
& = (v, K \nabla_x \phi)_{X_\Sigma} - (v, J \nabla_y \phi)_{X_\Sigma} .
\end{align*} 
Therefore, we can expand
\begin{align*}
\frac12 |\alpha|_{X_\Sigma}^2 & = \frac12 \int_\Sigma \big(|K \nabla_x \phi|^2 + |J \nabla_y \phi|^2 \big) \d x \d y - \int_\Sigma (J \nabla_x \phi, K \nabla_y \phi) \d x \d y \\ & = 
\frac12 \int_\Sigma \big(|\nabla_x \phi|^2 + | \nabla_y \phi|^2 \big) \d x \d y - \int_\Sigma \omega_I(\nabla_x \phi, \nabla_y \phi) \d x \d y,
\end{align*}
where in the last line we have used the fact that $JK = I$. 

This shows that 
\begin{align*}
S_{bosonic} & = \frac12 \int_{S \times \Sigma} |\d \phi|^2 \, \dvol \\ & = \frac12 \int_S \big(|\partial_t \phi|_{X_\Sigma}^2 + |\alpha|_{X_\Sigma}^2 \big) \d t + \int h \, \d t .
\end{align*}

The variation of the field $\psi$ with respect to the supercharge $Q_A$ is given by
\begin{align*}
\delta \psi^i & = - \partial_t \phi^i - \im \, \Gamma^i_{jk} \chi^j \psi^i - \alpha^i \\ 
& = - \partial_t \phi^i - \im \, \Gamma^i_{jk} \chi^j \psi^i - K \nabla_x \phi^i + J \nabla_y \phi^i \\ 
& = (I a(\nabla \phi))^i - \im \, \Gamma^i_{jk} \chi^j \psi^i  .
\end{align*} 
This agrees with the supersymmetry transformation in \eqref{eq:3dSUSY} as desired.
\end{proof}

Since $\alpha$ is the real part of a closed $(1, 0)$ form $\beta$, we obtain an $\cN=4$ supersymmetric mechanics into $X_\Sigma$. The zeros of $\beta$ are solutions to
\[a(\dbar \phi) = 0,\]
i.e. $(J - \im K) \dbar \phi = 0$. Applying $J$ and rearranging terms we obtain the Cauchy--Riemann equation
\[\d \phi\circ j = I\circ \d \phi.\]
In other words, the zeros of $\beta$ are $I$-holomorphic sections of $\fX_\Sigma\rightarrow \Sigma$.

Let us now suppose $(X, I, \Omega_I)$ admits the structure of a complex symplectic algebraic variety. Moreover, suppose the complexification $G_\C$ of $G$ and the bundle $P_\C = P\times^G G_\C\rightarrow \Sigma$ are algebraic. Then $\fX_\Sigma = P_\C\times^{G_\C} X\rightarrow \Sigma$ is also algebraic. By the results of \cite{GinzburgRozenblyum} we obtain a $(-1)$-shifted symplectic structure on the space $\Sect(\Sigma, \fX_\Sigma)$ of algebraic sections of $\fX_\Sigma\rightarrow \Sigma$.

\begin{proposal}
Choose an orientation data on $\Sect(\Sigma, \fX_\Sigma)$. Then the space of states in the 3d A-model is the cohomology
\[\R\Gamma(\Sect(\Sigma, \fX_\Sigma), P_{\Sect(\Sigma, \fX_\Sigma)})\]
of the perverse sheaf $P_{\Sect(\Sigma, \fX_\Sigma)}$. It admits a grading by the symplectic volume of the section with respect to $\omega_I$.
\end{proposal}

\begin{example}
Consider the case $G=\U(1)=\SO(2)$ and $X=\T^* Y$ for a smooth complex algebraic variety $Y$. Equip $X$ with a $\U(1)$-action given by scaling the cotangent fiber. The isomorphism $P\times^\U(1) \R^2\cong \Omega^1_\Sigma$ uniquely determines $P$, so that $P^\C$ is the $\C^\times$-bundle corresponding to the canonical bundle $K_\Sigma\rightarrow \Sigma$. In this case
\[\Sect(\Sigma, \fX_\Sigma)\cong \T^*[-1] \Map(\Sigma, Y).\]
The component of the space $\Map(\Sigma, Y)$ containing $\phi\colon \Sigma\rightarrow Y$ has virtual dimension
\[\dim_\Map = \int_\Sigma \phi^* c_1(Y) + \dim(Y)(1-g).\]
Therefore, using \cref{ex:cotangent} we get that the space of states in the 3d A-model into $\T^* Y$ is the shifted Borel--Moore homology
\[\rH^\BM_{\dim_\Map - \bullet}(\Map(\Sigma, Y)).\]
This answer was previously proposed in \cite{Nakajima}.
\end{example}

\section{GL twist of the 4d $\cN=4$ super Yang--Mills theory}
\label{sect:GL}

In this section we describe a compactification of the GL twist \cite{Marcus,KapustinWitten} of the 4d $\cN=4$ super Yang--Mills theory on a 3-manifold.

\subsection{Twisted super Yang--Mills theory}

Consider the following data:
\begin{itemize}
\item $G$ is a compact Lie group equipped with nondegenerate symmetric bilinear pairing $(-, -)$ on its Lie algebra.

\item A parameter $\theta\in\R$.
\end{itemize}

We may define the 4d $\cN=4$ super Yang--Mills theory given the above data. It admits a twist (known as the GL twist) which allows us to consider the theory on an arbitrary Riemannian 4-manifold. Let $M$ be a closed oriented Riemannian 4-manifold. The theory has the following fields:
\begin{itemize}
\item A principal $G$-bundle $P\rightarrow M$.
\item A connection $A$ on $P$.
\item A one-form $\phi\in\Omega^1(M, \ad P)$.
\item Sections $\sigma,\tilde{\sigma}\in\Gamma(M, \ad P)$.
\item Odd one-forms $\psi,\tilde{\psi}\in\Pi\Omega^1(M, \ad P)$.
\item Odd two-forms $\chi^\pm\in\Pi\Omega^2(M, \ad P)$, where $\chi^+$ is self-dual and $\chi^-$ anti self-dual.
\item Odd sections $\eta,\tilde{\eta}\in\Pi\Gamma(M, \ad P)$.
\end{itemize}

There are two commuting supersymmetry transformations $Q_l$, $Q_r$. We will consider the supercharge $Q = u Q_l + v Q_r$, where $u,v\in\R$. The supersymmetry transformation is \cite[Formulas (3.27), (3.28)]{KapustinWitten}
\begin{equation}
\begin{cases}
\delta A_\mu = \im u \psi_\mu + \im v \tilde{\psi}_\mu \\
\delta \phi_\mu = \im v \psi_\mu - \im u \tilde{\psi}_\mu \\
\delta \sigma = 0 \\
\delta \tilde{\sigma} = \im u\eta + \im v \tilde{\eta} \\
\delta \chi^+ = u(F - \frac{1}{2}[\phi\wedge \phi])^+ + v (\d_A \phi)^+ \\
\delta \chi^- = v(F - \frac{1}{2}[\phi\wedge \phi])^- - u(\d_A \phi)^- \\
\delta \eta = v\d^*_A \phi + u[\tilde{\sigma}, \sigma] \\
\delta\tilde{\eta} = -u\d^*_A \phi + v[\tilde{\sigma}, \sigma] \\
\delta \psi = u\d_A \sigma + v [\phi, \sigma] \\
\delta \tilde{\psi} = v \d_A \sigma - u[\phi, \sigma].
\end{cases}
\label{eq:KWSUSY}
\end{equation}

The bosonic part of the action is (see \cite[Section 3.4]{KapustinWitten})
\begin{equation}
S_{bosonic} = \int_M\dvol_M\left(\frac{1}{2}|F_\cA|^2 + \frac{1}{2}|\d^*_A\phi|^2 - \frac{1}{2}[\tilde{\sigma}, \sigma]^2 + (\d_A\tilde{\sigma}, \d_A\sigma) + ([\phi, \tilde{\sigma}], [\phi, \sigma])\right) + \frac{\im\theta}{16\pi^2} \int_M (F_A\wedge F_A).
\label{eq:KWaction}
\end{equation}

\subsection{Supersymmetric mechanics on the space of connections}

Let $N$ be a closed oriented Riemannian 3-manifold, $S$ a closed oriented Riemannian 1-manifold and set $M = N\times S$ with the product metric and orientation. Consider the affine space $\Conn^{\triv}_{G_\C}(N)$ of $C^\infty$ connections $\cA=A+\im\phi$ on the trivial principal $G_\C$-bundle over $N$. It is naturally a weakly K\"ahler manifold \cite{Corlette} with the K\"ahler form
\[\omega = \int_N (\delta \phi\wedge \star_N \delta A).\]

Let $G_N = \Map(N, G)$ be the Fr\'echet Lie group of smooth maps $N\rightarrow G$. Its Lie algebra $\g_N = C^\infty(N; \g)$ carries an invariant symmetric bilinear form
\[(v, w)_N = \int_N \dvol_N (v(x), w(x)),\]
where on the right we use the pairing on $\g$. The group $G_N$ acts on $\Conn_{G_\C}(N)$ preserving the K\"ahler structure with the moment map
\[\mu = \d_A\star_N \phi.\]
A principal $G_N$-bundle $P_N\rightarrow S$ is the same as a principal $G$-bundle $P\rightarrow N\times S$ trivializable along the fibers of $N\times S\rightarrow S$ by \cref{prop:bundlecorrespondence}.

There is a Chern--Simons functional
\[S_{CS} = \frac{1}{2}\int_N\left((\cA\wedge \d\cA) + \frac{1}{3}(\cA\wedge [\cA\wedge \cA])\right)\]
on $\Conn^{\triv}_{G_\C}(N)$ which defines a holomorphic function. Its imaginary part is given by
\[\Im S_{CS} = \int_N \left((\phi\wedge F_A) - \frac{1}{6}(\phi\wedge [\phi\wedge\phi])\right).\]
The differential of $S_{CS}$ defines a closed $(1, 0)$ one-form
\begin{align*}
\delta S_{CS} &= \int_N (\delta \cA\wedge F_\cA) \\
&= \int_N \left(\delta A\wedge \left(F_A-\frac{1}{2}[\phi\wedge \phi]\right) - \delta\phi\wedge \d_A \phi\right) + \im\int_N \left(\delta A\wedge \d_A\phi + \delta\phi\wedge \left(F_A-\frac{1}{2}[\phi\wedge \phi]\right)\right)
\end{align*}
on $\Conn^{\triv}_{G_\C}(N)$. Define
\begin{equation}
\alpha = \frac{u^2-v^2}{u^2+v^2}\delta \Re S_{CS} - \frac{2uv}{u^2+v^2} \delta\Im S_{CS} \label{eq:KWalpha}
\end{equation}
and
\begin{equation}
a = \frac{\theta}{8\pi^2}\delta \Re S_{CS}. \label{eq:KWa}
\end{equation}

Note that $\alpha$ is the real part of a closed $(1, 0)$-form
\[\beta = \frac{u^2-v^2+\im uv}{u^2+v^2} \delta S_{CS}\]
on $\Conn^{\triv}_{G_\C}(N)$ and similarly for $a$.

\begin{thm}
The GL twist of the 4d $\cN=4$ super Yang--Mills theory on $N\times S$ is equivalent to the $\cN=4$ $G_N$-gauged supersymmetric mechanics of maps $S\rightarrow \Conn^{\triv}_{G_\C}(N)$. Under this correspondence the supersymmetry transformation \eqref{eq:KWSUSY} corresponds to the transformation induced by the A supercharge $Q_A$ \eqref{eq:1dAgaugedN=4supercharge}.
\end{thm}
\begin{proof}
We decompose the one-form fields under the splitting $M=N\times S$ as follows:
\[A\mapsto A + A_0,\qquad \phi\mapsto \phi + \phi_0,\qquad \psi\mapsto \psi + \psi_0,\]
where the first component is a one-form along $N$ and the second component is a one-form along $S$. As before, we denote by $\cA = A + \im\phi$ the component of the complexified connection on $P$ in the $N$ direction; we denote by $\cA_0 = A_0 + \im\phi_0$ the component of the complexified connection on $P$ in the $S$ direction.

If $\star$ is the Hodge star operator on $N\times S$ and $\star_N$ is the Hodge star operator on $N$, then we have $\star(\d t\wedge \gamma) = \star_N \gamma$ for any one-form $\gamma$ along $N$, where $\d t$ is a one-form of norm $1$ along $S$. (Anti) self-dual two-forms on $M$ are identified with one-forms $\gamma$ along $N$: (anti) self-dual two-forms on $N\times S$ are $\d t\wedge \gamma \pm \star_N \gamma$. Therefore, we may write
\begin{equation}
\chi^\pm = \d t\wedge \chi^\pm_N \pm \star_N\chi^\pm_N.
\label{eq:selfdualcompactification}
\end{equation}

If we match the fields as in \cref{tbl:KWcompactification}, then the supersymmetry transformations \eqref{eq:KWSUSY} coincide with \eqref{eq:1dAgaugedN=4supercharge}.

The action \eqref{eq:KWaction} becomes
\begin{align*}
S_{bosonic} = \int_M\dvol_M&\Bigl(\frac{1}{2}|F_\cA|^2 + \frac{1}{2}|\d \cA_0 + \d_{\cA_0} \cA|^2 + \frac{1}{2}|\d^*_A\phi|^2 + \frac{1}{2}|\d_{A_0} \phi_0|^2 + (\d^*_A\phi, \d^*_{A_0} \phi_0) \\
- &\frac{1}{2}[\tilde{\sigma}, \sigma]^2 + (\d_A\tilde{\sigma}, \d_A\sigma) + (\d_{A_0}\tilde{\sigma}, \d_{A_0}\sigma) + ([\phi, \tilde{\sigma}], [\phi, \sigma]) + ([\phi_0, \tilde{\sigma}], [\phi_0, \sigma])\Bigr) \\
+ &\frac{\im\theta}{8\pi^2} \int_M F_A\wedge (\d A_0 + \d_{A_0} A).
\end{align*}

The norm squared of $\alpha$ with respect to the metric on $\Conn^{\triv}_{G_\CC}(N)$ is
\[|\alpha|^2 = \int_N \dvol_N |F_\cA|^2.\]
Also, observe that
\[
\frac{1}{2} \int_M \dvol_M(|\d \cA_0 + \d_{\cA_0} \cA|^2 + (\d^*_A\phi, \d^*_{A_0} \phi_0)) = \frac{1}{2}\int_M \dvol_M(|\d A_0 + \d_{A_0} A + \im\d_{A_0} \phi|^2 + |\d_A\phi_0 + \im[\phi, \phi_0]|^2).
\]

Using these identities it is easy to see that the action \eqref{eq:KWaction} is equivalent to \eqref{eq:1dN=4gaugedsigmamodel} substituting the fields using \cref{tbl:KWcompactification}.
\end{proof}

\begin{table}[h]
\begin{tabular}{cc}
\hline
supersymmetric mechanics & GL twist \\
\hline
$A$ & $A_0$ \\
$\eta$ & $-\frac{u\eta+v\tilde{\eta}}{u^2+v^2}$ \\
$c$ & $\frac{u\tilde{\eta}-v\eta}{u^2+v^2}$ \\
$\nu$ & $v\psi_0-u\tilde{\psi}_0$ \\
$\lambda$ & $u\psi_0 + v\tilde{\psi}_0$ \\
$\varphi$ & $-(u^2+v^2)\sigma$ \\
$\xi$ & $-\frac{1}{u^2+v^2} \tilde{\sigma}$ \\
$\sigma$ & $\phi_0$ \\
$\phi$ & $A + \im\phi$ \\
$\chi$ & $(u+\im v)(\psi-\im \tilde{\psi})$ \\
$\psi$ & $-\frac{(u+\im v)(\chi^+_N-\im \chi^-_N)}{u^2+v^2}$ \\
\hline
\end{tabular}
\caption{Fields in the $\cN=4$ gauged supersymmetric mechanics and in the GL twist of the 4d $\cN=4$ super Yang--Mills theory.}
\label{tbl:KWcompactification}
\end{table}

The zero locus of $\alpha$ on $\Conn^{\triv}_{G_\C}(N)/\Map(N, G_\C)$ coincides with the moduli space of flat $G_\C$-connections $\Loc^{\triv}_{G_\C}(N)$ on a trivializable $G_\C$-bundle. Let us also consider nontrivial $G_\C$-bundles. Assuming $G_\C$ is an algebraic group, there is a $(-1)$-shifted symplectic stack $\R\Loc_{G_\C}(N)$ parametrizing $G_\C$-local systems (see e.g. \cite{PTVV}) whose classical truncation is $\Loc_{G_\C}(N)$. It follows from \cite[Theorem 4.8]{JoyceTanakaUpmeier} and \cite{JoyceUpmeier1} that $\R\Loc_{G_\C}(N)$ carries a canonical orientation data (called a spin structure in those papers). Note that the restriction of $a$ to $\Loc_{G_\C}(N)$ is zero, so the twist by $\cL_a$ in \cref{mainproposalstack} disappears.

\begin{proposal}
Suppose $G_\C$ is an algebraic group. Then the space of states in the GL twist on a closed oriented 3-manifold $N$, for generic parameters, is
\[\R\Gamma(\Loc_{G_\C}(N), P_{\R\Loc_{G_\C}(N)}).\]
\end{proposal}

The previous complex was considered in \cite{AbouzaidManolescu} where it was called the \emph{complexified instanton Floer homology} of $N$.

\section{Haydys--Witten theory}
\label{sect:HaydysWitten}

In this section we describe a compactification of a topological twist of the 5d $\cN=2$ super Yang--Mills theory considered in \cite{WittenFivebranes} on a K\"ahler surface.

\subsection{Twisted super Yang--Mills theory}
Consider a compact Lie group $G$ equipped with a nondegenerate symmetric bilinear pairing $(-, -)$ on its Lie algebra. We may define the 5d $\cN=2$ super Yang--Mills theory with a gauge group $G$. It admits a topological twist introduced in \cite{WittenFivebranes} which allows us to consider the theory on the product $M\times S$ of an oriented Riemann 4-manifold $M$ and an oriented Riemannian 1-manifold $S$. We call it the \emph{Haydys--Witten} twist of the 5d $\cN=2$ super Yang--Mills theory. The coordinates along $M$ will have Greek indices and $S$ will have a coordinate $t$.

\begin{notation}
For a vector bundle $V\rightarrow M\times S$ we denote by $\Omega^n_M(M\times S, V)$ the space of $V$-valued $n$-forms on $M\times S$ along the $M$ direction.
\end{notation}

The theory has the following fields:
\begin{itemize}
\item A principal $G$-bundle $P\rightarrow M\times S$.
\item A connection $A$ on $P$ in the $M$ direction and $A_0$ in the $S$ direction.
\item A self-dual two-form $B\in\Omega^2(M\times S, \ad P)$.
\item Sections $\sigma,\tilde{\sigma}\in\Gamma(M\times S, \ad P)$.
\item Odd one-forms $\psi, \tilde{\psi}\in\Pi\Omega^1_M(M\times S, \ad P)$.
\item Odd self-dual two-forms $\chi, \tilde{\chi}\in\Pi\Omega^2_M(M\times S, \ad P)$.
\item Odd sections $\eta, \tilde{\eta}\in\Pi\Gamma(M\times S, \ad P)$.
\end{itemize}

We denote by $F_A\in\Omega^2_M(M\times S, \ad P)$ the curvature of $A$ and by $F_t\in\Omega^1_M(M\times S, \ad P)$ the contraction of the curvature of $A+A_0$ with $\partial_t$.

\begin{notation}
Suppose $X$ and $Y$ are self-dual two-forms on an oriented Riemannian 4-manifold. Then one can define their cross product to be (see \cite[Formula (5.29)]{WittenFivebranes} and \cite[Section 2.2]{QiuZabzine})
\[(X\times Y)_{\mu\nu} = X_{\mu \lambda} Y_\nu^{\ \lambda}-X_{\nu\lambda}Y_\mu^{\ \lambda}.\]
In the case $X,Y$ are self-dual two-forms valued in a bundle of Lie algebras, their cross product is symmetric and is given by
\[(X\times Y)_{\mu\nu} = \frac{1}{2}[X_{\mu \lambda}, Y_\nu^{\ \lambda}] - \frac{1}{2}[X_{\nu\lambda}, Y_\mu^{\ \lambda}].\]
\end{notation}

\begin{remark}
Let $N$ be an oriented Riemannian 3-manifold and consider the product metric on $N\times S^1$. Identify the self-dual two-forms on $N\times S^1$ with $\Omega^1_N(N\times S^1)$ using \eqref{eq:selfdualcompactification}. Then the cross product of the self-dual two-forms on $N\times S^1$ is identified with the cross product $X, Y\mapsto \star(X\wedge Y)$ of one-forms on $N$.
\end{remark}

The supersymmetry transformation is written in \cite[Section 4]{Anderson} (we take $u=1$ and $v=0$). Adjusting the conventions slightly (\cite{Anderson} works in the Minkowski signature, while we work in the Euclidean signature), it is
\begin{equation}
\begin{cases}
\delta A_\mu = \im\tilde{\psi}_\mu \\
\delta A_0 = \im\tilde{\eta} \\
\delta \sigma = -\im \sqrt{2}\eta \\
\delta\tilde{\sigma} = 0 \\
\delta B^{\mu\nu} = \im\tilde{\chi}^{\mu\nu} \\
\delta\eta = [\sigma, \tilde{\sigma}] \\
\delta\tilde{\eta} = -\sqrt{2}\d_{A_0}\tilde{\sigma} \\
\delta \psi = -F_t + \d^*_A B \\
\delta\tilde{\psi} = -\sqrt{2}\d_A \tilde{\sigma} \\
\delta \chi_{\kappa\lambda} = -2 (F_A)^+_{\kappa\lambda} + \frac{1}{2}(B\times B)_{\kappa\lambda} + \d_{A_0} B_{\kappa\lambda} \\
\delta \tilde{\chi}_{\kappa\lambda} = -\sqrt{2}[B_{\kappa\lambda}, \tilde{\sigma}].
\end{cases}
\label{eq:HWSUSY}
\end{equation}

The bosonic part of the action is (see \cite[Formula 5.40]{Anderson})
\begin{equation}
\begin{split}
S_{bosonic} = \int_{M\times S}\dvol_{M\times S}\Bigl(&\frac{1}{2}|F_A|^2 + \frac{1}{2}|F_t|^2 + \frac{1}{8}(\d_A B_{\mu\nu}, \d_A B^{\mu\nu}) + (\d_A \sigma, \d_A\tilde{\sigma}) \\
&+\frac{1}{4}|\d_{A_0} B|^2 + (\d_{A_0}\sigma, \d_{A_0}\tilde{\sigma}) \\
&+\frac{1}{16}|B\times B|^2 + \frac{1}{2}([B, \sigma], [B, \tilde{\sigma}]) - \frac{1}{2}|[\sigma, \tilde{\sigma}]|^2 \\
&+\frac{1}{8} R |B|^2 - \frac{1}{8}R_{\mu\nu\rho\sigma} B^{\mu\rho} B^{\nu\sigma}\Bigr).
\end{split}
\label{eq:HWaction}
\end{equation}
Here $R$ is the scalar curvature and $R_{\mu\nu\rho\sigma}$ is the Riemann curvature tensor.

Consider the Fr\'echet manifold $\Conn^{\triv,+}_G(M)$ parametrizing pairs of a connection $A$ on the trivial $G$-bundle $P\rightarrow M$ together with a self-dual $\ad P$-valued two-form $B$. It carries a metric
\[g(\delta A + \delta B, \delta A + \delta B) = \int_M \dvol_M (\delta A, \delta A) + \frac{1}{2}\int_M\dvol_M (\delta B, \delta B).\]
The group $G_M = \Map(M, G)$ acts on $\Conn^{\triv, +}_G(M)$ by gauge transformations on $A$ and by conjugation on $B$. This action preserves the metric.

Consider a smooth $G_M$-invariant function $f\colon \Conn^{\triv, +}_G(M)\rightarrow \R$ given by
\begin{equation}
f(A, B) = \int_M \dvol_M\left(-(F_A, B) +\frac{1}{12}(B\times B, B)\right).
\end{equation}
Its differential is
\begin{equation}
\alpha = \delta f = \int_M \dvol_M\left(-(F_A, \delta B) + \frac{1}{4}(B\times B, \delta B) + (\d_A^* B, \delta A)\right)
\label{eq:HWalpha}
\end{equation}

In addition, consider the $G_M$-invariant function $h\colon \Conn^{\triv, +}_G(M)\rightarrow \R$ given by the first Pontryagin class
\begin{equation}
h(A, B) = \frac{1}{2}\int_M (F_A\wedge F_A).
\label{eq:HWh}
\end{equation}

Note that this function is \emph{identically zero} since we are restricting to connections on topologically trivial $G$-bundles, but it has nontrivial values if we include nontrivial bundles.

\begin{thm}
The Haydys--Witten twist of the 5d $\cN=2$ super Yang--Mills theory on $M\times S$ is equivalent to the $\cN=2$ $G_M$-gauged supersymmetric mechanics of maps $S\rightarrow \Conn^{\triv, +}_G(M)$ with $\alpha$ given by \eqref{eq:HWalpha} and $h$ given by \eqref{eq:HWh}. Under this correspondence the supersymmetry transformation \eqref{eq:HWSUSY} corresponds to the transformation induced by the A supercharge $Q_A$ \eqref{eq:1dAgaugedN=2supercharge}.
\label{thm:HWRiemannian}
\end{thm}
\begin{proof}
We can match the fields in the super Yang--Mills theory and the gauged supersymmetric mechanics as shown in \cref{tbl:HWcompactification}. It is then straightforward to check that the A supersymmetry transformation in gauged supersymmetric mechanics corresponds to the supersymmetry transformation \eqref{eq:HWSUSY}.

Using integration by parts one may compute that
\[\int_M \dvol_M\left(|\d_A^* B|^2 - (F_A^+, B\times B)\right) = \int_M \dvol_M\left(\frac{1}{4}(\d_A B_{\mu\nu} \d_A B^{\mu\nu}) + \frac{1}{4} R|B|^2 - \frac{1}{4} R_{\mu\nu\rho\sigma} B^{\mu\rho}B^{\nu\sigma}\right).\]
Moreover,
\[|F_A^+|^2\dvol_M = \frac{1}{2}|F_A|^2\dvol_M + \frac{1}{2}(F_A\wedge F_A).\]
Combining these identities it is easy to see that the action \eqref{eq:HWaction} coincides with the action \eqref{eq:1dN=2gaugedsigmamodel} of gauged supersymmetric mechanics.
\end{proof}

\begin{table}[h]
\begin{tabular}{cc}
\hline
supersymmetric mechanics & Haydys--Witten twist \\
\hline
$A$ & $A_0$ \\
$\eta$ & $-\eta$ \\
$\lambda$ & $\tilde{\eta}$ \\
$\varphi$ & $\sqrt{2}\tilde{\sigma}$ \\
$\xi$ & $\frac{1}{\sqrt{2}}\sigma$ \\
$\phi$ & $A, B$ \\
$\chi$ & $\tilde{\psi}, \tilde{\chi}$ \\
$\psi$ & $\psi, -\chi$ \\
\hline
\end{tabular}
\caption{Fields in the $\cN=2$ gauged supersymmetric mechanics and in the Haydys--Witten twist of the 5d $\cN=2$ super Yang--Mills theory.}
\label{tbl:HWcompactification}
\end{table}

\subsection{K\"ahler case}

Suppose now that $M$ is a K\"ahler manifold. We will denote by $(-, -)_\C$ the hermitian extension of the metric on differential forms to complexified differential forms, by $(-, -)_\omega$ the corresponding symplectic structure and by $(-, -)$ the $\C$-linear extension of the metric.

In this section we are going to show that the supersymmetry of the compactified model is enhanced. Consider the Fr\'echet manifold $\Conn^{\overline{\partial}, \triv, (2, 0)}_{G_\C}(M)$ parametrizing $(0, 1)$ connections $A_{0, 1}$ on the trivial $G_\C$-bundle over $M$ together with an $\ad P$-valued $(2, 0)$ form $B_{2, 0}$. It will also be convenient to identify $A_{0, 1}$ with the $(0, 1)$ part of a connection $A$ on the trivial $G$-bundle over $M$. This manifold admits a linear K\"ahler structure associated with the Hermitian metric
\[(\delta A_{0, 1} + \delta B_{2, 0}, \delta A'_{0, 1} + \delta B'_{2, 0})_\C = 2\int_M\dvol_M (\delta A_{0, 1}, \delta A'_{0, 1})_\C + \int_M \dvol_M (\delta B_{2, 0}, \delta B'_{2, 0})_\C.\]

$\Conn^{\overline{\partial}, \triv, (2, 0)}_{G_\C}(M)$ admits a holomorphic action by $\Map(M, G_\C)$ given by a gauge transformation on $A_{0, 1}$ and conjugation on $B_{2, 0}$. The following lemma is proven by a straightforward computation.

\begin{lm}
For $X,Y\in\Omega^{2, 0}(M)$ we have
\[\omega\cdot (X, Y)_\omega = -2 \Re X\times \Re Y.\]
\end{lm}

Using this identity we can see that the subgroup $G_M=\Map(M, G)\subset \Map(M, G_\C)$ of compact gauge transformations acts by isometries and admits a moment map
\begin{equation}
\mu(A_{0, 1}, B_{2, 0}) = -(F_A - \Re B_{2, 0}\times \Re B_{2, 0})\wedge \omega.
\label{eq:HWmoment}
\end{equation}

Consider the holomorphic function $W\colon \Conn^{\overline{\partial}, \triv, (2, 0)}_{G_\C}(M) \rightarrow \C$ given by
\begin{equation}
W(A_{0, 1}, B_{2, 0}) = -2\int_M \dvol_M (F_{A_{0, 1}}, B_{2, 0})
\end{equation}
and its differential
\begin{equation}
\beta = \partial W = -2\int_M \dvol_M \left((\delta A_{0, 1}, \partial_{A_{0, 1}}^* B_{2, 0}) +(F_{A_{0, 1}}, \delta B_{2, 0})\right).
\label{eq:HWbeta}
\end{equation}

\begin{thm}
Suppose $M$ is a K\"ahler surface. The Haydys--Witten twist of the 5d $\cN=2$ super Yang--Mills theory on $M\times S$ is equivalent to the $\cN=4$ $G_M$-gauged supersymmetric mechanics of maps $S\rightarrow \Conn^{\overline{\partial}, \triv, (2, 0)}_{G_\C}(M)$ with $\beta$ given by \eqref{eq:HWbeta} and $h$ given by \eqref{eq:HWh}. Under this correspondence the supersymmetry transformation \eqref{eq:HWSUSY} corresponds to the transformation induced by the A supercharge $Q_A$ \eqref{eq:1dAgaugedN=4supercharge}.
\end{thm}
\begin{proof}
By \cref{thm:HWRiemannian} the compactification is equivalent to the $\cN=2$ $G_M$-gauged supersymmetric mechanics of maps $S\rightarrow \Conn^{\triv, +}_G(M)$. We may identify self-dual two-forms on a Kahler manifold as
\[\Omega^{2, +}(M)\cong \Omega^{2, 0}(M)\oplus \Omega^0(M)\cdot\omega,\]
where the map $\Omega^{2, +}(M)\rightarrow \Omega^{2, 0}(M)$ is given by taking the $(2, 0)$ component of a self-dual two-form and the decomposition on the right-hand side is orthogonal. This allows us to identify
\[\Conn^{\triv, +}(M)\cong \Conn^{\overline{\partial}, \triv, (2, 0)}_{G_\C}(M)\times \Map(M, \g)\]
by sending $(A, B)$ to $A_{0, 1}$ given by the $(0, 1)$ component of $A$, $B_{2, 0}$ given by the $(2, 0)$ component of $B$ and $B_\omega\in\Map(M, \g)$ given by the $\omega$ component of $B$. Under this decomposition
\[f(A, B) = \Re W(A_{0, 1}, B_{2, 0}) + \int_M \mu(A_{0, 1}, B_{2, 0}) B_\omega.\]
The claim then follows from \cref{rmk:N=2toN=4}.
\end{proof}

The critical points of $W$ on $\Conn^{\overline{\partial}, \triv, (2, 0)}_{G_\C}(M)$ are given by
\begin{equation}
F_{A_{0, 1}} = 0,\qquad \partial^*_{A_{0, 1}} B_{2, 0} = 0.
\label{eq:HWcriticalpoints}
\end{equation}
Using the K\"ahler identities we may rewrite the last equation as $\overline{\partial}_{A_{0, 1}} B_{2, 0} = 0$.

Let us now assume $M$ is a projective surface and $G_\C$ is an algebraic group. Then we may consider the derived algebraic stack $\R\Bun_{G_\C}(M)$ of algebraic $G_\C$-bundles on $M$ which is quasi-smooth. The shifted cotangent stack $\T^*[-1]\R\Bun_{G_\C}(M)$ parametrizes principal $G_\C$-bundles $P\rightarrow M$ together with an algebraic section $B\in\Gamma(M, \ad P\otimes K_M)$, i.e. solutions of \eqref{eq:HWcriticalpoints}. Using \cref{ex:cotangent} we arrive at the following.

\begin{proposal}
Suppose $G_\C$ is an algebraic group and $M$ a projective surface. Then the space of states in the Haydys--Witten twist on $M$ is the shifted Borel--Moore homology
\[\rH^{\BM}_{\dim(\R\Bun_{G_\C}(M))-\bullet}(\Bun_{G_\C}(M))\]
of the moduli stack of $G_\C$-bundles on $M$. It has a natural grading given by the second Chern character $\int_M\ch_2(P)$ of the $G_\C$-bundle $P$.
\end{proposal}

\section{The twist of the 7d $\cN=1$ super Yang--Mills theory}
\label{sect:G2monopoles}

In this section we describe a compactification of the topological twist of 7d $\cN=1$ super Yang--Mills theory on a Calabi--Yau three-fold. 

\subsection{Twisted super Yang--Mills theory}
Consider a compact Lie group $G$ equipped with a nondegenerate symmetric bilinear pairing $(-, -)$ on its Lie algebra and consider the 7d $\cN=1$ super Yang--Mills theory with a gauge group $G$. It admits a topological twist (see e.g. \cite{AOS}) which allows us to consider the theory on a $G_2$ manifold $M$.

\begin{notation}
Denote the fundamental $3$-form on the $G_2$ manifold $M$ by $\varphi \in \Omega^3(M)$.
\end{notation}

The theory has the following fields:
\begin{itemize}
\item A principal $G$-bundle $P \to M$.
\item A connection $A$ on $P$.
\item Three sections $\sigma, \rho, \Tilde{\rho} \in \Gamma(M , \ad P)$. 
\item Two odd one-forms $\psi,\chi \in \Pi \Omega^1(M , \ad P)$. 
\item Odd sections $\nu, \eta \in \Pi \Gamma(M , \ad P)$.
\end{itemize}

The supersymmetry transformation is
\begin{equation}
\begin{cases}
\delta A  = \im \, \chi \\
\delta \sigma  = \im \, \nu \\
\delta \rho  = 0 \\
\delta \Tilde{\rho}  = 2 \im \, \eta \\
\delta \eta = \frac{1}{2} \, [\rho, \Tilde{\rho}] \\
\delta \chi  = \d_A \rho \\
\delta \nu = [\sigma, \rho] \\
\delta \psi  = \d_A \sigma - \star ( \star \varphi \wedge F_A) .
\end{cases}
\label{eqn:7susy}
\end{equation}

The bosonic part of the action is
\begin{equation}
\begin{split}
S_{bosonic} = \int_{M}\dvol_M \, \bigg(\frac{1}{2}|F_A|^2 + \frac{1}{2}|\d_A \sigma|^2 + (\d_A \rho, \d_A\tilde{\rho}) + ([\sigma, \rho], [\sigma, \tilde{\rho}]) - \frac{1}{2}|[\rho, \tilde{\rho}]|^2 \bigg)
\end{split}
\label{eq:7dbosonaction}
\end{equation}

\begin{remark}
The above formulas are obtained by a dimensional reduction from the formulas in \cite[Section 3]{AOS} which describe the topological twist of the 8-dimensional super Yang--Mills theory on an $8$-manifold with ${\rm Spin}(7)$-holonomy.
\end{remark}

\subsection{Calabi--Yau compactification}

Suppose that $X$ is a smooth projective Calabi--Yau 3-fold with a holomorphic volume form $\Omega \in \Omega^{3,0}(X)$ and K\"ahler form $\omega_X \in \Omega^{1,1}(X)$. Our convention is that
\[\dvol_X = \frac{1}{4}\Re\Omega\wedge \Im\Omega = \frac{\omega_X^3}{6}.\]
Let $S$ be a one-dimensional Riemannian manifold and consider the product Riemannian metric on $M=X \times S$.
There is a natural $G_2$ structure on $X \times S$ with the fundamental three-form
\[
\varphi = \Re(\Omega) - \d t \wedge \omega_X
\]
and the fundamental four-form
\[\star\varphi = -\d t\wedge \Im(\Omega) - \frac{\omega_X^2}{2}.\]

Consider the Fr\'echet manifold $\Conn^{\triv}_G(X)$ parametrizing connections $A$ on the trivial $G$-bundle on $X$. It admits a Riemannian structure
\[g(v, w) = \int_X \dvol_X \, (v, w) .\]
The group $G_X = \Map(X, G)$ acts on $\Conn^{\triv}_G(X)$ by gauge transformations on $A$. 
Its Lie algebra $\fg_X = C^\infty(X ; \fg)$ carries an invariant symmetric bilinear form 
\[
(v, w)_{\fg_X} = \int_X \dvol_X \, (v(x), w(x)) .
\]
This action preserves the metric.

$\Conn^{\triv}_G(X)$ is equipped with a K\"ahler structure with K\"ahler form defined by
\[
\omega(v, w) = \frac12 \int_X \omega_X^2 \wedge (v \wedge w) .
\]
The $G_X$ action is Hamiltonian with respect to this K\"ahler structure. The moment map is
\[
\mu = \frac{1}{2} \omega_X^2 \wedge F_A = (\Lambda F_A) \dvol_X .
\]

Consider the holomorphic Chern--Simons functional
\[S_{hCS} = -\frac{\im}{2}\int_X \Omega \wedge \left(A\wedge \d A + \frac{1}{3}(A\wedge [A\wedge A])\right)\]
which defines a holomorphic function on $\Conn^{\triv}_G(X)$. (Notice that $S_{hCS}$ depends holomorphically on the $(0,1)$ part of the connection $A$.) 
Define the one-form $\alpha$ on $\Conn^{\triv}_G(X)$ by
\begin{equation}\label{eqn:7alpha}
\alpha = \delta \Re S_{hCS} = \int_X \Im(\Omega) \wedge \delta A \wedge F_A .
\end{equation}
Finally, define the smooth function
\begin{equation}\label{eqn:7h} 
h = -\frac{1}{2}\int_X (F_A\wedge F_A)\wedge \omega.
\end{equation}

\begin{thm}
Let $X$ be a Calabi--Yau 3-fold.
The topological twist of the 7d $\cN=1$ super Yang--Mills theory on $X \times S$ is equivalent to the $\cN=4$ $G_X$-gauged supersymmetric mechanics of maps $S\rightarrow \Conn^{\triv}_G(X)$ with $\alpha$ given by \eqref{eqn:7alpha} and $h$ given by \eqref{eqn:7h}. Under this correspondence the supersymmetry transformation \eqref{eqn:7susy} corresponds to the transformation induced by the A supercharge $Q_A$ \eqref{eq:1dAgaugedN=4supercharge}.
\label{thm:7cy}
\end{thm}
\begin{proof}
We decompose the one-form fields under the splitting $M=X\times S$ as follows:
\[
A\mapsto A + A_0, \quad \psi \mapsto \psi + \psi_0 , \quad \chi \mapsto \chi + \chi_0,
\]
where the first component is a one-form along $X$ and the second component is a one-form along $S$. 
For instance, $A$ is a connection on $P$ in the $X$ direction and $A_0$ is a connection on $P$ in the $S$ direction. We may then match the fields as in \cref{tbl:7dcpt}.

The action \eqref{eq:7dbosonaction} becomes 
\begin{align*}
S_{bosonic} = \int_M\dvol_M&\Bigl(\frac{1}{2}|F_A|^2 + \frac{1}{2}|\d A_0 + \d_{A_0} A|^2 + \frac{1}{2}|\d_{A} \sigma|^2 + \frac12 |\d_{A_0} \sigma|^2 \\
+ &(\d_A \rho, \d_A \Tilde{\rho}) + (\d_{A_0} \rho, \d_{A_0} \Tilde{\rho}) + (\d_{A_0}\tilde{\sigma}, \d_{A_0}\sigma) + ([\sigma, \rho], [\sigma, \Tilde{\rho}]) - \frac12 |[\rho, \Tilde{\rho}]|^2 \Bigr).
\end{align*}

We have
\begin{align*}
\frac{1}{2}|\alpha|^2 &= \frac{1}{2} \int_X |\Im\Omega\wedge F_A|^2\dvol_X \\
&= \frac{1}{2}\int_X \left|\frac{\Omega}{2\im}\wedge F_{0, 2} - \frac{\Bar{\Omega}}{2\im}\wedge F_{2, 0}\right|^2\dvol_X \\
&= 2 \int_X (F_{0, 2}, F_{2, 0}) \dvol_X
\end{align*}
and
\[\frac{1}{2}|\mu|^2 = \frac{1}{2}\int_X |\Lambda F_A|^2 \dvol_X.\]
By \cite[Corollary 4.3]{ESW} we have
\[\frac{1}{2}|F_A|^2\dvol_X + \frac{1}{2}(F_A\wedge F_A)\wedge\omega = \left(2(F_{2, 0}, F_{0, 2}) + \frac{1}{2}(\Lambda F_A)^2\right)\dvol_X.\]
This shows that the bosonic action agrees with that of the $\cN=4$ SUSY mechanics \eqref{eq:1dN=4gaugedsigmamodel}.

We now show that the reduction of the 7d supersymmetry matches that of the supersymmetric mechanics. 
The variation of the field $c$ in supersymmetric mechanics is obtained from the variation of the 7d field $\psi_0$. 
We read this off as
\[
\delta \psi_0 = \d_{A_0} \sigma + \frac12 \star_7 (\omega^2_X \wedge F_A) .
\]
(The subscript in $\star_7$ is to emphasize that we are using the Hodge operator with respect to the metric on the $7$-manifold $X \times S$.)
Using $c = -\psi_0\d t$ we see this matches with the supersymmetry in \eqref{eq:1dAgaugedN=4supercharge}.
The variation of the field $\psi$ in supersymmetric mechanics is obtained from the variation of the 7d field $\psi$. We read this off as
\[
\delta \psi = \d_A \sigma + \star_7 \left(\d t \wedge \Im (\Omega) \wedge F_A + \frac{\omega_X^2}{2}\wedge (\d A_0 + \d_{A_0} A)\right).
\]
This coincides with the variation in \eqref{eq:1dAgaugedN=4supercharge}.
\end{proof}

\begin{table}[h]
\begin{tabular}{cc}
\hline
SUSY mechanics & 7d theory \\
\hline
$A$ & $A_0$ \\
$\eta$ & $\eta$ \\
$c$ & $- \psi_0\d t$ \\
$\nu $ & $\nu$ \\
$\lambda $ & $\rho$ \\
$\xi$ & $\Tilde{\rho}$ \\
$\sigma$ & $\sigma$ \\
$\phi$ & $A$ \\
$\chi$ & $\chi$ \\
$\psi$ & $-\psi\d t$ \\
\hline
\end{tabular}
\caption{Fields in the $\cN=4$ gauged supersymmetric mechanics and in the twist of the 7d $\cN=1$ super Yang--Mills theory.}
\label{tbl:7dcpt}
\end{table}

We may identify connections $A$ on the trivial $G$-bundle over $X$ with $(0, 1)$ connections $A_{0, 1}$ on the trivial $G_\C$-bundle over $X$. The action of $\Map(X, G)$ extends to an action of $\Map(X, G_\C)$ by complexified gauge transformations. The zero locus of $\alpha$ on $\Conn^{\triv}_G(X)/\Map(X, G_\C)$ coincides with the moduli space $\Bun_G^{\triv}(X)$ of holomorphic structures on a trivializable $G_\C$-bundle. Let us also consider nontrivial $G_\C$-bundles. If $G_\C$ is an algebraic group, there is a $(-1)$-shifted symplectic stack $\R\Bun_{G_\C}(X)$ parametrizing algebraic $G_\C$-bundles (see e.g. \cite{PTVV}). It is shown in \cite{JoyceUpmeier2} that $\R\Bun_{G_\C}(X)$ carries a canonical orientation data for $G=\SU(n)$.

\begin{proposal}
Suppose $G_\C$ is an algebraic group. Then the space of states in the topological twist of the 7d $\cN=1$ super Yang--Mills theory on a smooth projective Calabi--Yau 3-fold $X$ is
\[\R\Gamma(\Bun_{G_\C}(X), P_{\R\Bun_{G_\C}(X)}).\]
It carries a natural grading by the second Chern character $\int_X \ch_2(P)\wedge \omega$ of the $G_\C$-bundle $P$.
\end{proposal}

The previous complex gives the \emph{categorified Donaldson--Thomas invariants} of $X$; this definition was introduced in the papers \cite{BBDJS,KiemLi}.

\printbibliography

\end{document}